\documentclass[a4paper, 11pt]{article}
\usepackage[utf8]{inputenc}
\usepackage{amsmath}
\usepackage{amssymb}
\usepackage{amsthm}
\usepackage{xcolor}
\usepackage{mathtools}
\usepackage{dsfont}

\usepackage{hyperref}
\hypersetup{colorlinks,citecolor=blue,filecolor=blacklinkcolor=red, urlcolor=magenta,linkcolor=red}
\newcommand{\dd}{\textrm{d}}
\newcommand{\N}{\mathbb{N}}
\newcommand{\Z}{\mathbb{Z}}
\newcommand{\R}{\mathbb{R}}
\newcommand{\C}{\mathbb{C}}

\DeclareMathOperator{\sh}{sh}
\DeclareMathOperator{\ch}{ch}
\DeclareMathOperator{\argch}{argch}
\newcommand{\gO}[3]{\underset{{#2}\to{#3}}{O}\left({#1}\right)} 
\newcommand{\po}[3]{\underset{{#2}\to{#3}}{o}\left({#1}\right)} 
\newtheorem{defi}{Definition}
\newtheorem{lemm}{Lemma}
\newtheorem{coro}{Corollary}
\newtheorem{thm}{Theorem}
\newtheorem{prop}{Proposition}
\theoremstyle{remark}
\newtheorem{rema}{Remark}
\newtheorem{innerassu}{Assumption}
\newenvironment{assu}[1]
  {\renewcommand\theinnerassu{#1}\innerassu}
  {\endinnerassu}

\usepackage[articlein=false,giveninits=true,style=ext-alphabetic,hyperref,backend=biber]{biblatex}
\AtEveryBibitem{%
  \clearfield{issn} 
  \clearfield{doi} 

  \ifentrytype{online}{}{
    \clearfield{url}
  }
}

\title{Inverse Regge poles problem on a warped ball} 
\author{Jack Borthwick\thanks{University Bourgogne Franche-Comté, Laboratoire de mathématiques de Besançon, UMR CNRS 6623, Université de Franche-Comté, 16, route de Gray, 25030 Besançon cedex, \href{mailto:jack.borthwick@math.cnrs.fr}{jack.borthwick@math.cnrs.fr}}, Nabile Boussaid\thanks{Laboratoire de mathématiques de Besançon, UMR CNRS 6623, Université de Franche-Comté, 16, route de Gray, 25030 Besançon cedex, \href{mailto:nabile.boussaid@univ-fcomte.fr}{nabile.boussaid@univ-fcomte.fr}}, Thierry Daudé\thanks{Laboratoire de mathématiques de Besançon, UMR CNRS 6623, Université de Franche-Comté, 16, route de Gray, 25030 Besançon cedex, \href{mailto:thierry.daude@univ-fcomte.fr}{thierry.daude@univ-fcomte.fr}}}\date{}

\begin{document}
\maketitle

\begin{abstract}
In this paper, we study a new type of inverse problem on warped product Riemannian manifolds with connected boundary that we name warped balls. Using the symmetry of the geometry, we first define the set of Regge poles as the poles of the meromorphic continuation of the Dirichlet-to-Neumann map with respect to the complex angular momentum appearing in the separation of variables procedure. These Regge poles can also be viewed as the set of eigenvalues and resonances of a one-dimensional Schrödinger equation on the half-line, obtained after separation of variables. Secondly, we find a precise asymptotic localisation of the Regge poles in the complex plane and prove that they uniquely determine the warping function of the warped balls.  
\end{abstract}


\section{Introduction}

\subsection{The model of warped balls and the statement of the inverse problem} 

This paper is devoted to the study of an inverse problem on warped product manifolds with connected boundary from a new set of spectral data that we refer to as \emph{Regge poles}. 

Precisely, let $K$ be a compact $(n-1)$-dimensional ($n\geq 2$) Riemannian manifold with metric $g_K$ and consider the warped product: \[\overline{M}=(0,1]\times K,\] with interior $M=(0,1)\times K$ and boundary $\partial M = \{1\} \times K$, equipped with the metric:
\[g = c^2(r)\left(\dd r^2+ r^2g_K \right),\] for some function $c>0$. Setting $x=-\ln r$ brings the metric $g$ into the form:
\[g=c(e^{-x})^2e^{-2x}\left( \dd x^2 + g_K \right)\equiv f(x)^2g_0, \]
so that $g$ is conformally equivalent to the product metric:
\[ 
g_0 = \dd x^2+g_K,
\]
on $[0,+\infty) \times K$. The conformal factor is defined for $x>0$ by \[f(x)=c(e^{-x})e^{-x}>0.\]
We emphasise that under these general assumptions, the above metrics are not necessarily regular, but can have a conical singularity at $r=0$. 
Actually, this is, in some sense, the generic situation as one proves\footnote{see \cite[Section 4.3.4]{Pe2016}} that the metric $g$ is regular if and only if:
\begin{enumerate}\item the odd-order derivatives $c^{(2k+1)}(0)$ vanish for every $k\in \N$, \item $K = \mathbb{S}^{n-1}$ and $g_{K}= d\Omega^2$ where $d\Omega^2$ is the round metric on $\mathbb{S}^{n-1}$. \end{enumerate}

Our ultimate goal is to determine the conformal factor $c$ (or equivalently $f$) from the knowledge of the \emph{Regge poles} (see Definition~\ref{defi:regge_poles}) in addition to some boundary datum. 
For simplicity we shall make the assumption that the conformal factor $c$ is a small perturbations of $1$, in the sense that:
\begin{equation} \label{hyp:cf1} \begin{cases} c= 1 + \tilde{V}>0,\tilde{V}\in L^{\infty}_{\textrm{comp}}((0,1]), \\ \exists d\in \R_+^*, \tilde{V} > -1 +d. \end{cases}\tag{CF1}\end{equation}
Or, in terms of $f$: 
\begin{equation} \label{hyp:cf'} \begin{cases} f= e^{-x}+V>0, V\in L^{\infty}_{\textrm{comp}}(\mathbb{R}_+),\\ \exists d\in \R_+^*, e^xV > -1 +d. \end{cases} \tag {CF1'}\end{equation}
We shall require some minimal regularity in the sense that:
\begin{equation}\label{hyp:cf2}  V\in C^1_c(\R_+), V'\in W^{1,\infty}(\R_+) \tag{CF2}.\end{equation}

The starting point towards defining the Regge poles is to consider the \emph{Dirichlet-to-Neumann} operator, a natural quantity in the study of inverse problems. In fact, we shall consider a slight generalisation of this allowing for fixed non-zero energy $\lambda>0$. Consider first, for any $\phi\in H^{\frac{1}{2}}(\partial M)$, the (non-homogenous) Dirichlet boundary condition problem in the natural Hilbert space: $L^2(M)$:
\begin{equation} \label{eq:dirichlet_problem} \begin{cases} -\Delta_g u = \lambda u \,\, \textrm{on $M$}, \\ \phantom{-}u=\phi \,\, \textrm{on $\partial M$}.\end{cases}\end{equation}
Recall that: $L^2(M)=L^2(\R_+\times K, f^n \dd x \dd K)$.

Under the assumptions~(\ref{hyp:cf'},~\ref{hyp:cf2}), we will show that given any $\phi \in H^{\frac{1}{2}}(\partial M)$ there is a unique solution $u_\phi$ to problem~\eqref{eq:dirichlet_problem} in $H^1(M)$ and define the Dirichlet-to-Neumann operator (at fixed energy $\lambda)$ to be the operator: \[\Lambda(\lambda): H^{\frac{1}{2}}(\partial M) \rightarrow H^{-\frac{1}{2}}(\partial M),\] mapping any $\phi \in H^{\frac{1}{2}}(\partial M)$ to the boundary value of the (outgoing) normal derivative of the solution to problem~\eqref{eq:dirichlet_problem}, in other words:
\[ \Lambda(\lambda) \phi = -\partial_x u_\phi(0, \cdot). \]

Since $g$ is conformally equivalent to the product metric, the Laplacian $\Delta_g$ is closely related to the product Laplacian $\Delta_{g_0}$. In fact, the change of variable $v=f^{\frac{n}{2}-1}u$, shows that problem~\eqref{eq:dirichlet_problem} is equivalent to:
\begin{equation} \label{eq:dirichlet_problem2} \begin{cases} -\Delta_{g_0}v +q_f v = \lambda f^2 v, \, x>0, \\ \phantom{-}v(0,\cdot)=f^{\frac{n}{2}-1}(0)\phi.\end{cases}\end{equation} 
with $v\in L^2(\mathbb{R}_+\times K , f^2\dd x\dd K)=L^2(\mathbb{R}_+,f^2\dd x )\otimes L^2(K)\footnote{Because $f$ depends only on $x$.}$. The potential $q_f$ is given by:
\begin{equation} q_f = \frac{(f^{\frac{n}{2}-1})''}{f^{\frac{n}{2}-1}}.  \end{equation}

The symmetry of the overall manifold can now be used to its full advantage. Indeed, the Laplacian $-\Delta_K$ on the compact manifold $K$ has compact resolvent and $L^2(K)$ can be decomposed onto an orthonormal basis of eigenvectors $(Y_k)_{k\in \mathbb{N}}$. Let $(\mu_k^2)_{k\in \N}$ denote the eigenvalues of $-\Delta_K$, counted with multiplicity and ordered such that $0=\mu_0^2<\mu_1^2\leq \mu_2^2 \leq \dots $. 
Recall now that:
\[ \Delta_{g_0}=\partial_x^2 + \Delta_K.\]
Since the coefficients of $-\Delta_{g_0}v + q_f -\lambda f^2$ depend only on the coordinate $x$, the operator is stable on any of the subspaces: \[E_k \equiv L^2(\mathbb{R}_+, f^2\dd x)\otimes\textrm{span}~\{Y_k\}.\] On each $E_k$ it reduces to a $1$-dimensional Schrödinger operator and \eqref{eq:dirichlet_problem2} can be written in terms of the components of the decompositions:
\[\phi = \displaystyle \sum_{k\in \mathbb{N}} \phi_kY_k, \quad v=\displaystyle \sum_{k\in \mathbb{N}} v_k Y_k,\]
as:
\begin{equation} \label{DirichletProblem1D}\begin{cases} -v_k'' + (q_f - \lambda f^2) v_k= -\mu_k^2 v_k, \\ v_k(0)=f^{\frac{n}{2}-1}(0)\phi_k. \end{cases}\end{equation}
A key point is that introducing:
\[ \begin{gathered} Q_f = q_f - \frac{(n-2)^2}{4}-\lambda(f-e^{-x})(f+e^{-x}), \\ z_k^2= \mu_k^2 + \frac{(n-2)^2}{4}, \end{gathered}\] the equation on each component can be rewritten as the Schrödinger equation on the halfline:
\begin{equation}\label{eq:schro_v} -v'' + (Q_f - \lambda e^{-2x}) v= -z^2 v,\end{equation}
where we have dropped, for convenience, explicit dependence on the index $k$. Assumptions~(\ref{hyp:cf'},~\ref{hyp:cf2}) ensure that the potential $Q_f \in L^\infty_{\textrm{comp}}(\R^+)$. The Schrödinger operator: \[H = -\frac{d^2}{dx^2} + Q_f(x) - \lambda e^{-2x},\] with Dirichlet boundary conditions at $x=0$ will have a central importance in this work; it will be referred to as the \textit{associated Schrödinger operator}.

Lastly, we point out that our spectral study will put emphasis on the spectral parameter $-z^2$ appearing in \eqref{eq:schro_v} as opposed to the more natural parameter $\lambda$ which will be fixed throughout. In particular, $(-z^2)$ will be not be restricted to the discrete values $-z_k^2$ and we will be especially interested in what happens for complex values. From this perspective, it can be naturally viewed as a \emph{complex angular momentum} coming from the separation of variables. 

Let us now take a closer look at the Dirichlet-to-Neumann operator in terms of $v$. It is also stable on any of the subspaces $E_k$ and its restriction can be written:
\[\Lambda^k(\lambda) \phi_k = -u_k'(0)= -\frac{u'_k(0)}{u_k(0)}\phi_k = -\left[ \frac{v'_k(0)}{v_k(0)} - \left(\frac{n}{2}-1\right)\frac{f'(0)}{f(0)}\right] \phi_k.\]

Let $\phi \in H^{\frac{1}{2}}(M)$ as in the previous paragraph, in Section~\ref{sec:exist_unique} we will show that when $n\geq 3$ the components $v_k$ of $v=f^{\frac{n}{2}-1}u_{\phi}$ satisfy: $v_k, v'_k \in L^2(\R^+)$. Hence, the term $\frac{v'_k(0)}{v_k(0)}$ is in fact the value of the \emph{Weyl-Titchmarsh} function for the Dirichlet Schrödinger operator $H$ at $-z_k^2$. We will use the usual notation\footnote{A way one can understand this notation is as follows, let $\mathcal{S}$ be the Riemann surface $\mathcal{S}=\{(w,z)\in\C^2, w+z^2=0\}$.  $\mathcal{S}$ is in fact diffeomorphic to $\mathbb{C}$ and the variable $z$ is a \emph{global} coordinate. The Weyl-Titchmarsh function can first be defined as a function on the open set $U=\mathcal{S}\cap\{(z,w)\in \mathbb{C}^2, \Re z >0\}$, on which we could use the variable $w=-z^2$ as local coordinate. The notation $M(-z^2)$ denotes the value of $M$ defined on $U$ at the point $(-z^2,z)\in U$. We will show that $M$ has a meromorphic extension to all of $\mathcal{S}\cong \C$, and, in order to avoid confusion and simplify notation in later analysis, we will write $m(z)$ to denote the value of the meromorphic extension of $M$ at any point $(-z^2,z)$ of $\mathcal{S}$. This is equivalent to expressing $M$ in the global coordinate chart given by $z$. }:
\[
  \frac{v_k'(0)}{v_k(0)} = M(-z_k^2),
\]  
for the Weyl-Titchmarsh function associated to \eqref{eq:schro_v}. 

In summary, the Dirichlet-to-Neumann operator $\Lambda(\lambda)$ can thus be diagonalized onto the Hilbert basis of eigenvectors $(Y_k)_{k \in \N}$ and on each invariant subspace $E_k$, it acts as an operator of multiplication by:
\[
  \Lambda^k(\lambda) = \Lambda(\lambda, z_k) =  -\left[ M(-z_k^2) - \left(\frac{n}{2}-1\right)\frac{f'(0)}{f(0)}\right].
\]
Note that $\Lambda(\lambda)$ is completely determined by the Weyl-Titchmarsh function $M$, the shifted eigenvalues $z_k^2$ of the transversal Laplacian $-\Delta_K$ and the quotient $\frac{f'(0)}{f(0)}$.
 
We will prove that in our particular model, the Weyl-Titchmarsh function $M(-z^2)$ is meromorphic on $\C$. By analogy with scattering potential problems \cite{Re1959, AlRe1965}, we then define:
\begin{defi}\label{defi:regge_poles} The \textbf{Regge poles} are the poles of the meromorphic extension to $\mathbb{C}$ of the Weyl-Titchmarsh function $M(-z^2)$ of the operator $H$. \end{defi} It is well known (see for instance \cite{RaSi2000}) that these poles can be equivalently defined as the set of Dirichlet eigenvalues and resonances of the operator $H$. They can therefore be thought of as resonances with respect to the shifted ``angular momentum'' $z$.

The question we want to address in the present work is then the following: \emph{does the knowledge of the Regge poles  of a warped ball $(M,g)$ determine uniquely the warping function $c$ (or $f$)} ? 

According to the previous observation, this problem amounts to studying an \emph{inverse resonance problem} for the half-line Schrödinger operator $H$. Note that, in the case of compactly supported potentials, such problems have been studied thoroughly in \cite{Ko2004, MSW2010}. One of the novelties of this work consists in considering \emph{non-compactly supported potentials} of the form: 
\[
  Q_f - \lambda e^{-2x}, \quad Q_f \in L^\infty_{\textrm{comp}}(\R^+).
\]

\subsection{The results and strategy of the proof} 

Our results on the Regge poles problem are twofold. First, we give a positive answer to the uniqueness inverse Regge poles problem:
\begin{thm} \label{MainInverse}
If two warping factors $f$ and $\tilde{f}$ satisfying~\eqref{hyp:cf'} and~\eqref{hyp:cf2} have the same Regge poles $\{(\alpha_k)_{k \geq 0} \cup (\beta_j)_{j \in \Z^*} \}$ then:
\[Q_f=Q_{\tilde{f}}.\] 
Hence, assuming $f(0)=\tilde{f}(0)$ and $f'(0)=\tilde{f}'(0)$, we have:
\[ f=\tilde{f}.\]
\end{thm}

The second part of our results pertain to the distribution of the Regge poles. However, we need to make another assumption on the warping function $f = e^{-x} + V$. Precisely, we assume additionally that:
\begin{assu}{(CF3)}\phantom{a}
\label{cf3}
\begin{itemize}
\item supp\,$V = [0,a]$, 
\item There exists $p \in \N^*$, such that $V \in C^{p}(\R^+) \cap AC^{p+1}(\R^+)$, 
\item $\partial_x^{p+1} V$ is continuous $\R^+$ except at $a$ where it has distinct left and right limits.
\end{itemize} 
\end{assu}
Assumption~\ref{cf3} entails that the corresponding potential $Q_f$ satisfies: \[\textrm{supp}\,Q_f = [0,a], \quad Q_f \in C^{p-2}(\R^+) \cap AC^{p-1}(\R^+),\] and $\partial_x^{p-1} Q_f$ is continuous on $\R^+$ except at $a$ where it has a jump. (If $p=1$, we will understand this to mean that $Q_f$ is continuous on $\R^+$ except at $a$ where it has a jump.)
 
With this assumption, we obtain precise asymptotics of the Regge poles. Precisely, we have: \begin{thm} \label{MainAsymptotic}
The set of Regge poles is given by the union of two sequences 
\[
  \{(\alpha_k)_{k \geq 0} \cup (\beta_j)_{j \in \Z^*} \}, 
\]
where: 
\begin{enumerate}
\item The $\alpha_k$'s are real numbers asymptotically close to the negative integers $-k$ as $k \to \infty$, \textit{i.e.} for any $0< \delta <\frac{1}{2}$, there exists $N$ such that  for all $k \geq N$, 
$$
   |\alpha_k + k | < \delta.
$$
\item The $\beta_j$'s are complex numbers which form a set that is symmetric with respect to the real axis and satisfies for all $j >> 1$, 
\[
\begin{aligned}
\beta_{\pm j} = \pm i \frac{\pi}{2a} \left( 2|j| + \frac{p-1}{2} \pm (\textrm{sgn}(A) + 1)   \right) &- \frac{p+1}{2a} \log \frac{|j| \pi}{a} \\ &+ \frac{1}{2a} \log |A| (p-1)! + o(1),
\end{aligned}
\]
for a certain explicit constant $A$.
\end{enumerate}
\end{thm}

Roughly speaking, the first sequence of Regge poles $(\alpha_k)$ is due to the exponentially decreasing potential $-\lambda e^{-2x}$ in the Schrödinger operator $H$, whereas the second, $(\beta_j)$, is due to the compactly supported perturbation $Q_f$.

\begin{rema} 
We emphasize that the \emph{large} Regge poles $(\alpha_k)_{k \geq 0}$ seem to be quite stable under compactly supported perturbations. We mean by this that if we replace $Q_f \in L^\infty([0,a])$ by another potential $\widetilde{Q_f} \in L^\infty([0,\tilde{a}])$, then the corresponding $\tilde{\alpha}_k$'s remain asymptotically close to the negative integers $-k$. 

In opposition with the above corollary, the Regge poles $(\beta_j)$ are quite unstable under compactly supported perturbation. This can be seen from the precise asymptotics given in Theorem \ref{MainAsymptotic}. Indeed, slight modifications  in the support $[0,a]$ of the perturbation, or in the value of the jump of $Q_f$ at $a$, entail dramatic changes in the asymptotics. This is a well-known phenomenon for scattering resonances \cite{DyZw2019} and black holes quasinormal modes \cite{DMBJ2021}.  	
\end{rema}

Theorem~\ref{MainInverse} will be proved according to two distinct strategies, the first will not require assumption~\ref{cf3} and will provide a general uniqueness result. The second will avail of~\ref{cf3} in order to follow the strategy of~\cite{BKW2003, BW2004}. This will enable us to obtain a formula for the Weyl-Titchmarsh function in terms of the Regge poles:
\begin{thm}\label{thm:wt_expression}
Under the assumptions given in the text, and assuming, for readability that all the Regge poles are simple then:
\[
  M(-z^2) = -z + \sum_{k \geq 0} \frac{a_k}{z - \alpha_k} + \sum_{j \in \Z^*} \frac{b_j}{z - \beta_j}.
\]
Similar expressions exists when the poles are not simple.
\end{thm}
The strategy is then as follows:
\begin{itemize}
\item \textbf{Step 1}: We consider the Jost solution $\psi(x,z)$ of the one-dimensional Schrödinger equation, that is the unique solution (up to multiplicative constant) of \eqref{eq:schro_v} that is $L^2$ at $x=+\infty$. The Weyl-Titchmarsh function $M(-z^2)$ is then defined for $\Re z >0$ by:
$$
  M(-z^2) = \frac{\psi'(0,z)}{\psi(0,z)}. 
$$
When $V = 0$, the Jost function $\psi(0,z)$ turns out to be explicitly given by a Bessel function. When $V \ne 0$, using the notion of transformation operators introduced by Marchenko (see for instance the presentation given in \cite{Le2018}), and a perturbation argument, we will show that the map $z \mapsto \psi(0,z)$ is an entire function of order $1$ and infinite type. Moreover, we obtain precise asymptotics of $\psi(0,z)$ for $|z|$ large in the complex plane. These asymptotics allow us to locate the zeros $\psi(0,z)$ by a standard Rouché argument and classical results on the location of zeros of Laplace transforms due to Hardy, Cartwright \cite{Ca1930, Ca1931, Ha1905} (we refer to Zworski \cite[p. 287]{Zw1987} for a similar application of these results). The zeros of $\psi(0,z)$ being the Regge poles of $(M,g)$, Theorem \ref{MainAsymptotic} will be proved.  
\item \textbf{Step 2}: Thanks to Step 1, the Weyl-Titchmarsh function $M(-z^2)$ has a meromorphic extension to the Riemann surface: \[ \{ (w,z)\in \C^2, w+z^2=0\}\cong \C,\] with poles given by the Regge poles. Using a Cauchy theorem on a well-chosen contour, we are able to express this function as in Theorem~\ref{thm:wt_expression}:
$$
  M(-z^2) = -z + \sum_{k \geq 0} \frac{a_k}{z - \alpha_k} + \sum_{j \in \Z^*} \frac{b_j}{z - \beta_j}, 
$$
where, in the above, we assume that all Regge poles are simple for readability.
Moreover, using essentially the Hadamard factorisation theorem for entire functions of finite order \cite{Le1996}, we can prove that the residues $a_k$ and $b_j$ only depend on the Regge poles $\{(\alpha_k)_{k \geq 0} \cup (\beta_j)_{j \in \Z^*} \}$. Consequently, we prove that the Weyl-Titchmarsh function $M$ is uniquely determined by the Regge poles. However, it is well-known that the potential $Q_f$ is uniquely determined from the Weyl-Titchmarsh function by the Borg-Marchenko theorem \cite{Be2001, GeSi2000, Si1999}. 

\end{itemize}

\subsection{Some bibliographical comments}

In this section, we give some references that solve inverse problems on warped product Riemannian manifolds similar to the one studied in this paper. 

In \cite{DKN2021, Ge2020, Ge2022}, it is shown that the Steklov spectrum, \textit{i.e.} the eigenvalues of the Dirichlet-to-Neumann map $\Lambda(0)$ at frequency $0$,  of a warped product with connected and disconnected boundary uniquely determines the warping function $f$. Moreover, logarithmic stability estimates are provided.  

In \cite{IsKo2017, IsKo2021}, some inverse spectral problems on compact and non-compact rotationally symmetric manifolds are studied. These are similar to the warped product Riemannian manifolds considered in this paper. In \cite{IsKo2017} for instance, it is shown that the spectral data consisting in the eigenvalues and norming constants of a family of one-dimensional Sturm-Liouville operators,  corresponding to the diagonalisation of the Laplace-Beltrami operator onto a fixed spherical harmonic, determines uniquely the \emph{compactly supported} potential. Moreover, the authors give an analytic isomorphism from the space of spectral data onto the space of functions describing the warping function, hence answering the inverse characterisation problem. In \cite{IsKo2021}, a similar uniqueness inverse result is obtained for a class of non-compact warped product Riemannian manifolds. Here, the spectral data are given by the set of eigenvalues and resonances of a family of one-dimensional Schrödinger operators, corresponding to the diagonalisation of the Laplace-Beltrami operator onto a fixed spherical harmonic. 

The results in \cite{IsKo2021} are the closest in spirit to the ones presented in this paper. The main difference with our work is that the authors only consider Schrödinger operators with \emph{compactly supported} potentials and their spectral data does not correspond to the Regge poles considered here.

\subsection{Content of the paper}

In Section~\ref{Preliminaries}, we provide some additional results on the geometric models of warped product Riemannian manifolds. In particular, we give the proof of the existence and uniqueness of a solution $u$ to the Dirichlet problem \eqref{eq:dirichlet_problem}. In Section~\ref{Unperturbed}, we study in detail the unperturbed model, that is the case $V=0$. The Jost function is shown to be merely a Bessel function and its behavior with respect to the complex angular momentum $z$ is studied. In Section~\ref{Perturbed}, we extend the analytic properties of the Jost function from the unperturbed case to the perturbed case. The main technique used here consists in using the so-called transformation operators that connect the solutions of the unperturbed Schrödinger equation to the solutions of the perturbed Schrödinger equation. Eventually, in Section~\ref{InversePb}, we put together all the previous results and prove our main theorems.

\section{Preliminary results} \label{Preliminaries}

\subsection{The relationship between $\Delta_g$ and $\Delta_{g_0}$}
To begin this section, we explain the equivalence between problems~\eqref{eq:dirichlet_problem} and~\eqref{eq:dirichlet_problem2} in a geometric manner by means of the conformal Laplacian. 
For this, recall that, in terms the coordinate $x=-\ln r \in \R^+$, warped-balls are described as $\overline{M}=\R^+ \times K$ equipped with the metric: 
\[g= f(x)^2 g_0, \]
where $f(x)=c(e^{-x})e^{-x}>0$ and $g_0 =  dx^2 + g_K$. 
$g$ and $g_0$ are therefore in the same conformal class.

For any $w \in \mathbb{R}$, let $\mathcal{E}[w]$ denote the module of sections of the bundle of conformal densities\footnote{A conformal density of weight $w$ can be thought of as a function on $M$ depending on the metric and homogeneous in the sense that: $f(x,\Omega^2 g)=\Omega^w f(x,g).$ Since any metric $g$ on $M$ defines a canonical volume density $\textrm{Vol}_g$, and that $\textrm{Vol}_{\Omega^2 g}=\Omega^n\textrm{Vol}_g$, one has a correspondence (that depends on the choice of metric) between the usual 1-density bundle and conformal densities of weight $-n$ via the map: $\phi \mapsto \frac{\phi}{\textrm{Vol}_g}$.} of weight $w$. The conformal Laplace operator is the operator acts from $\mathcal{E}[1-\frac{n}{2}]$ into $\mathcal{E}[-1-\frac{n}{2}]$ and can be calculated using any metric $g$ in a given conformal class according to the formula:
\[ \mathbf{g}^{ab}\nabla^g_a\nabla^g_b -\frac{(n-2)}{4(n-1)}\mathbf{g}^{ab}R^g_{ab}.\]
In this expression $R^g_{ab}$ is the Ricci tensor for the given metric $g$ and $\mathbf{g}$ is the \emph{conformal metric} $\mathbf{g}=(\sigma_g)^2g$ where $\sigma_g$ is the unique density of weight $1$ that evaluates to $1$ along the section $g$; unlike $g$; $\mathbf{g}$ is conformally invariant.

Using the conformal invariance of the conformal Laplacian in the specific case $g=f^2g_0$ and using bold letters to indicate that the conformal metric is used in contractions as opposed to a given metric in the conformal class, we find that for any conformal density $\sigma \in \mathcal{E}[1-\frac{n}{2}]$:
\[\mathbf{\Delta}_g\sigma = \mathbf{\Delta}_{g_0} \sigma -\frac{(n-2)}{4(n-1)}(\mathbf{R}_{g_0}-\mathbf{R}_g)\sigma.\]
Now:
\[\mathbf{R}_g-\mathbf{R}_{g_0}=\sigma_{g_0}^{-2}\left[(4-n)(n-1)\left(\frac{f'(r)}{f(r)}\right)^2  -2(n-1)\frac{f''(r)}{f(r)}.\right],\]
hence:
\[\begin{aligned} \mathbf{\Delta}_g\sigma&=\mathbf{\Delta}_{g_0}\sigma -\sigma^{-2}_{g_0}\left[\frac{(n-4)(n-2)}{4}\left(\frac{f'(r)}{f(r)}\right)^2 +\frac{n-2}{2}\frac{f''(r)}{f(r)} \right]\sigma,\\ &=\mathbf{\Delta}_{g_0}\sigma -\sigma^{-2}_{g_0}\left(\frac{(f^{\frac{n}{2}-1})''}{f^{\frac{n}{2}-1}} \right)\sigma.\end{aligned}\]
To reinterpret this relation in terms of functions set\footnote{$\sigma_{g}$ is parallel for the Levi-Civita connection of $g$ and similarly for $\sigma_{g_0}$.} $\sigma = u \sigma_{g}^{1-\frac{n}{2}}=v \sigma_{g_0}^{1-\frac{n}{2}}.$
Note that: \[ v=f^{\frac{n}{2}-1}u.\]
In particular, for any $\lambda \in \mathbb{R}_+$:
\[ \begin{aligned} -\Delta_g u = \lambda u &\Leftrightarrow -\mathbf{\Delta}_g \sigma=\lambda\sigma_{g}^{-2}\sigma, \\& \Leftrightarrow -\mathbf{\Delta}_{g_0}\sigma +\sigma^{-2}_{g_0}\left(\frac{(f^{\frac{n}{2}-1})''}{f^{\frac{n}{2}-1}} \right)\sigma = \lambda f^{2} \sigma_{g_0}^{-2}\sigma, \\&\Leftrightarrow -\Delta_{g_0}v + \left(\frac{(f^{\frac{n}{2}-1})''}{f^{\frac{n}{2}-1}} \right)v = \lambda f^2 v.  \end{aligned} \]


\subsection{A few properties of Bessel functions} \label{sec:bessel}
Bessel functions will play an important role in the sequel and we will need a few results regarding their behaviour as functions of \emph{their order}. We state and prove the relevant properties here. We first recall an integral formula~\cite[\textit{Schläfi's formula}, Equation~(4), \S 6.2 p.176 ]{watson1995treatise}:
\[\forall z\in \mathbb{C},  t\in \mathbb{R}_+^*, J_z(t)= \frac{1}{2\pi} \int_{-\pi}^\pi e^{-iz \theta}e^{i t \sin \theta} \textrm{d}\theta -\frac{\sin(\pi z)}{\pi}\int_0^\infty e^{-sz}e^{-t\sh s}\dd s, \] 
from which it is straightforward to see that $\psi_0: z\mapsto J_z(t)$ is an entire function. 
 In fact, from this representation we will show that $\psi_0$ is of finite order, this is contained in the following lemmata:
\begin{lemm}
\label{lemm:ordre_bessel}
For any compact interval $I=[\alpha,\beta]\subset \R_+^*$, there are constants $c_1,c_2>0$ such that: \[ |J_z(t)| \leq c_1e^{c_2z\ln(z)}, \]
uniformly in $t\in I$, for $|z|$ large enough. 
\end{lemm}
\begin{proof}
Since for any $\theta \in [-\pi,\pi]$, $|e^{-iz \theta}e^{it\sin\theta}| = e^{\Im z \theta} \leq e^{|z|\pi}$, the first integral is bounded by $e^{|z|\pi}$. Therefore it will satisfy an estimate of the required form for $|z|$ large enough. 
The second part of the formula can be estimated roughly by:
\[ \frac{e^{|z|\pi}}{\pi}\int_0^\infty e^{s|z| -t\sh s} \dd s. \]
It remains to determine the behaviour of the integral $\displaystyle \int_0^\infty e^{s|z| -t\sh s} \dd s$ for large values of $|z|$.
For this we notice first that when $|z| \geq t $, the map $s\mapsto s|z| - t \sh s$ has a global maximum at: \[s_\textrm{max}(|z|)=\argch\frac{|z|}{t}=\ln\left( \frac{|z|}{t} + \sqrt{\frac{|z |^2}{t^2}-1} \right).\] Its maximum value is given by: 
\[\phi_{\textrm{max}}(|z|)=|z| \argch\frac{|z|}{t} - \sqrt{|z|^2-t^2}.\]
%
Now:
\[ \int_0^\infty e^{s|z| -t\sh s} \dd s= e^{2\phi_{\max}(|z|)} \int_0^\infty e^{s|z|-2\phi_\textrm{max}(|z|)} e^{-t\sh s} \dd s  \]
For fixed $s>0$, the map $|z|\mapsto s|z|-2\phi_{\textrm{max}}(|z|)$ has a global maximum attained when $|z|=t\ch\frac{s}{2}$, so that for any $s>0$ and any $|z|\geq t>0$:
\[s|z|-2\phi_{\textrm{max}}(|z|) \leq 2t\sh \frac{s}{2}. \]
Hence:  \[ \int_0^\infty e^{s|z| -t\sh s} \dd s \leq e^{2\phi_\textrm{max}(|z|)}\int_0^{\infty}e^{t(2\sh\frac{s}{2}-\sh s)} \dd s,\]
the integral on the right converges and this estimate concludes the proof.
\end{proof}
\begin{rema}
In fact for any $\varepsilon >0$, there is $C>0$, such that one has:
\[ \int_0^\infty e^{s|z| -t\sh s} \dd s \leq C e^{(1+\varepsilon)\phi_\textrm{max}(|z|)},\]
\end{rema}

\begin{lemm}
\label{lemm:asymptotic_bessel}
Let $0 <\delta < \frac{1}{2}$, et let $\displaystyle U_{\delta}= \mathbb{C} \setminus \bigcup_{n=1}^{+\infty} D(-n,\delta)$, then uniformly for $t\in [\alpha,\beta]\subset \R_+^*$:
\begin{equation}\label{eq:fa} J_{z}(t) = \frac{\left(\frac{t}{2}\right)^{z}}{\Gamma(z+1)}\left( 1 + \gO{\frac{1}{|z|}}{|z|}{\infty} \right). \tag{FA}\end{equation}
\end{lemm}

\begin{proof} 
Recall the more usual series representation of the Bessel function for $t>0$ given by:
\[ J_z(t)=\left(\frac{t}{2}\right)^{z} \sum_{m=0}^{+\infty} \frac{(-1)^m}{\Gamma(m+1)\Gamma(z+m+1)}\left(\frac{t}{2}\right)^{2m}. \]
Using that $\Gamma(z+1)=z\Gamma(z)$, this can be rewritten:
\[ J_z(t)=\frac{\left(\frac{t}{2}\right)^z}{\Gamma(z+1)}\left(1 +\sum_{m=1}^{+\infty} \frac{(-1)^m}{m!(z+1)_{(m)}} \left(\frac{t}{2}\right)^{2m} \right), \]
where $(x)_{(n)}$ denotes the Pochhammer symbol, or increasing factorial: \[ (x)_{(n)}=x(x+1)\dots(x+n-1).\]
On $U_\delta$ one certainly has: \[ (z+1)_{(m)} \geq |z+1|\delta^{m-1}, \]
so that: \[\left\lvert\sum_{m=1}^{+\infty} \frac{(-1)^m}{m!(z+1)_{(m)}}\left(\frac{t}{2}\right)^{2m}\right\rvert \leq \frac{\delta}{|z+1|}e^{\frac{t^2}{4\delta}}.\qedhere\]
\end{proof}

From these two lemmata it now follows that:
\begin{coro}
For fixed $t>0$, the entire function $z\mapsto J_z(t)$ is of finite order $\rho=1$.
\end{coro}

The asymptotics~\eqref{eq:fa} are fundamental and will be used extensively in what follows.

\subsection{Gelfand-Levitan transformation operators}\label{sec:transfo}
As previously mentioned, the operator:
\[ H= -\frac{d^2}{dx^2} - \lambda e^{-2x} + Q_f, \]
is of central interest in this work. This is clearly a perturbation of the Schrödinger operator:
\[ \begin{gathered} H_0= -\frac{d^2}{dx^2}  + q_0, \quad q_0(x)=-\lambda e^{-2x}, x\in \R_+. \end{gathered} \]
In this section we will write: $q=q_0+ Q_f.$
Our assumptions on $f$ ensure that the perturbation $Q_f=q-q_0$ is bounded and of compact support. For definiteness we assume:
\[ \textrm{supp}~Q_f \subset [0,a], a\in \mathbb{R}_+^* \]
Finally,  let $D(H_0)$ (resp. $D(H)$) denote the maximal domain in $L^2(\R_+,\dd x)$ of $H_0$ (resp. $H$):
\[D(H_0) = \{ g \in L^2(\R_+,\dd x), g,g' \in AC(\R_+), H_0g \in L^2(\R_+,\dd x) \},\]
and similarly for $H$.

Since $H$ is a very reasonable perturbation of $H_0$ solutions of $Hg=-z^2g$ can be related to those of $H_0g=-z^2g$ by means of a \emph{transformation operator}~\cite[Chapter 1]{Le2018}.  This is in all points analogous to a scattering operator but in a slightly different context. 

%
The result is that one can find an integral kernel $K$ such that the operator $X$ defined by:
\[ (Xf)(x) = f(x) + \int_x^{+\infty} K(x,t)f(t)\dd t, x >0,\]
maps $D(H_0)$ into $D(H)$ and satisfies the interlacing relationship: \[ HX=XH_0.\] 
%
%
%
%
The kernel $K$ is defined as a fixed point:
\[
K(x,t)=\frac{1}{2}\int_\frac{t+x}{2}^\infty Q_f(s)\dd s+\int_\frac{t+x}{2}^\infty \int_0^\frac{t-x}{2}(q(\alpha-\beta)-q_0(\alpha+\beta))K(\alpha-\beta,\alpha+\beta)\dd \beta \dd \alpha, 
\]
for $0\leq x\leq t$.

We will now reformulate and refine a little the results of~\cite[Section 5]{BW2004}. Below we will omit the dependence on $f$ but emphasise on the dependence in $Q$. We assume:
\[
M_{q,q_0}:=\sup\left\{ \int_0^\frac{t-x}{2}|q(\alpha-\beta)-q_0(\alpha+\beta)|\dd \beta,\; \frac{t-x}{2}\leq \alpha\leq a,\, 0\leq x\leq t\right\}<\infty.
\]

\begin{rema}
We can consider $K_Q$ continuous on $\R^+\times \R^+$ with the requirement
\[
 K_Q(x,t)=0 ~\textrm{ if } x+t \geq 2a \textrm{ or }  x\geq t.
\]
Note however that the condition that $K_Q$ has support on $x\leq t$ does not seem to be fulfilled by the term $\int_\frac{t+x}{2}^\infty Q(s)\dd s$ while it is for the others as, by convention, the integrals vanish whenever the lower bound is larger than the upper one. So from this point of view, below we could consider $\frac{1}{2}{\mathds{1}}_{ x\leq t}\int_\frac{t+x}{2}^\infty Q(s)\dd s$ instead. 

We will actually consider $K_Q$ continuous on $\Omega=\{ (x,t)\in \R^2, 0\leq x\leq t\}$ with the requirement
\[
 K_Q(x,t)=0 ~\textrm{ if } x+t \geq 2a.
\]
The proof will be the same and up to a multiplication by $\mathds{1}_{ x\leq t}$ which only affects the first term, the solutions will coincide.
\end{rema}

Let $E_a$ be the set of 
continuous 
functions 
on $\Omega$
with support in $\Omega_0=\{(x,t)\in \R^+\times \R^+,\,0\leq x \leq t,\, x+t\leq 2a\}$. This space is endowed with the supremum norm. We consider the map $\mathcal{F}_Q$:
\[
 \mathcal{F}_Q(L)(x,t):=\frac{1}{2}\int_\frac{t+x}{2}^\infty Q(s)\dd s+\int_\frac{t+x}{2}^\infty \int_0^\frac{t-x}{2}(q(\alpha-\beta)-q_0(\alpha+\beta))L(\alpha-\beta,\alpha+\beta)\dd \beta\dd \alpha.
\]
If $Q$ is integrable then \underline{$E_a$ is stable under ${\mathcal F}_Q$}.
\begin{rema}
Note that $K_Q$ is the unique fixed point of $\mathcal{F}_Q$ if and only if $K_Q$ is the unique fixed point of $\mathcal{F}_Q^{n}$ for some $n\in\N$. 
\end{rema}

\begin{lemm}\label{Lem:contraction}
 Let $a>0$ and and $Q$ integrable with support in $[0,a]$. Let $q:=q_0+Q$. Then  there exists $N>0$ such that for $n\in\N$ and $n>N$, $\mathcal{F}_Q^{n}$ is a contraction mapping on $E_a$.
\end{lemm}
\begin{proof}
One can write:
\begin{align*}
  \mathcal{F}_Q^n(L)&=\sum_{k=1}^n K_{Q,n}
  +\mathcal{G}_{Q,n}(L),
\end{align*}
with:
\begin{align*}
\mathcal{G}_{Q,1}(L)(x,t)&:=\int_\frac{t+x}{2}^\infty \int_0^\frac{t-x}{2}(q(\alpha-\beta)-q_0(\alpha+\beta))L(\alpha-\beta,\alpha+\beta)\dd \beta\dd \alpha,\\
\mathcal{G}_{Q,n+1}(L)&:=
\mathcal{G}_{Q,1}(\mathcal{G}_{Q,n}(L))
=\mathcal{G}_{Q,1}^{n+1}(L),\\
K_{Q,1}(x,t)&:=\frac{1}{2}\int_\frac{t+x}{2}^\infty Q(s)\dd s,\\ 
K_{Q,n+1}&:=\mathcal{G}_{Q,1}(K_{Q,n})=\mathcal{G}_{Q,n}(K_{Q,1}),
\end{align*}
with the convention that $\mathcal{G}_{Q,0}$ is the identity.
Moreover, for $L\in E_a$, we have:
\begin{equation}\label{Eq:contraction}
 |\mathcal{G}_{Q,n}(L)(t,x)|\leq \frac{M_{q,q_0}^n}{n!}\left(a-\frac{t+x}{2}\right)_+^n\|L\|_\infty.
\end{equation}
This is proved by induction on $n$ following for instance \cite[Lemma 1]{BW2004}.
Hence, for $n$ sufficiently large, $\mathcal{G}_{Q,n}$ is a contraction on $E_a$ endowed with supremum norm $L^\infty$ and so is $\mathcal{F}_Q^n$ 
\end{proof}

From Banach's fixed point theorem, $\mathcal{F}_Q^n$ and $\mathcal{F}_Q$ have a unique fixed point on $E_a$ which is given by the convergent series in $E_a$: 
$$K_Q=\sum_{n\in\N} K_{Q,n},$$
defined in the proof of Lemma~\ref{Lem:contraction}.
Notice that the smaller $Q$ is, the smaller $K_Q$ is, and this ensures the invertibility of $X$.

Below for $O$ an open subset of $\R^d$ with closure $S$, $d\in \N$, $C^k(S)$ is the set of functions $F$ with are $C^k(O_F)$ for some open set of $\R^d$ with $S\subset O_F$. This is equivalent to the fact that $F$ belongs to $C^k(O)$ and its derivatives of order $k$ have a limit at any point of $\partial O$.
\begin{lemm} 
Let 
$Q\in L_a^1$.
Let $k\in \N$. If  $Q$ is in $C^k(\R^{+})$  
 then $K_Q$ is in $C^{k+1}(\Omega)$ and (even if $k=0$) the derivatives
 $$\partial_{\frac{t+x}{2}}\partial_{\frac{t-x}{2}}K_Q,\quad \partial_{\frac{t-x}{2}}\partial_{\frac{t+x}{2}}K_Q,$$
 exist and are both equal to $(q(x)-q_0(t))K_Q(x,t)$ and 
 \begin{equation}\label{Eq:Wave}
(\partial_t^2 K_Q-\partial_x^2K_Q)(x,t)=(q(x)-q_0(t))K_Q(x,t),
\end{equation}
in a distributional sense on the interior of $\Omega$ if $k=0$, or, in the strong sense for any $(x,t)\in \Omega$ if $k\in\N$.

If, moreover, $\partial_x^k Q$ is continuously differentiable except at $a$ where it has distinct right and left derivatives then: 
\[
(x,t)\mapsto H_Q(x-t):=K_Q(x,t)-\frac{1}{2}\int_\frac{t+x}{2}^\infty Q(s)\dd s,
\]
is $C^{k+2}(\Omega)$ and for any $x\in [0,2a]:$
$$\partial_x^{k+2}K_Q(x^-,(2a-x)^-)=-\frac14\partial_x^{k+1} Q(a^-).$$
\end{lemm}
\begin{proof}
Let us first consider $Q\in L_a^1$.
For $L\in E_a$, due to continuity of translations in $L^1$, the quantity:
\[
 \int_0^\frac{t-x}{2}\left(q\left(\frac{t+x}{2}-\beta\right)-q_0\left(\frac{t+x}{2}+\beta\right)\right)L\left(\frac{t+x}{2}-\beta,\frac{t+x}{2}+\beta\right)\dd \beta,
\]
is continuous in $\beta$ and absolutely continuous in $x+t$ whilst:
\[
 \int_\frac{t+x}{2}^\infty \left(q\left(\alpha-\frac{t-x}{2}\right)-q_0\left(\alpha+\frac{t-x}{2}\right)\right)L\left(\alpha-\frac{t-x}{2},\alpha+\frac{t-x}{2}\right)\dd \alpha,
\]
is continuous in $\alpha$ and absolutely continuous in $t-x$.

From the fundamental theorem of calculus, we deduce that
\begin{equation}\label{Eq:derivative1}
\partial_{\frac{t+x}{2}} \mathcal{G}_{Q,1}(L)(x,t)= -\int_0^\frac{t-x}{2}(q(\frac{t+x}{2}-\beta)-q_0(\frac{t+x}{2}+\beta))L(\frac{t+x}{2}-\beta,\frac{t+x}{2}+\beta)\dd \beta
\end{equation}
and, using first Fubini's theorem, we deduce similarly:
\begin{equation}\label{Eq:derivative2}
\partial_{\frac{t-x}{2}}\mathcal{G}_{Q,1}(L)(x,t)=
\int_\frac{t+x}{2}^\infty (q(\alpha-\frac{t-x}{2})-q_0(\alpha+\frac{t-x}{2}))L(\alpha-\frac{t-x}{2},\alpha+\frac{t-x}{2})\dd \alpha.
\end{equation}
In a distributional sense on the interior $\Omega$ we obtain:
\begin{equation}\label{Eq:derivative3}
\partial_{\frac{t-x}{2}}\partial_{\frac{t+x}{2}}\mathcal{G}_{Q,1}(L)(x,t)=
\partial_{\frac{t+x}{2}}\partial_{\frac{t-x}{2}}\mathcal{G}_{Q,1}(L)(x,t)=
(q(x)-q_0(t))L(x,t).
\end{equation}
It follows then that $K_Q$ satisfies the wave equation:
\[
(\partial_t^2 K_Q-\partial_x^2K_Q)(x,t)=(q(x)-q_0(t))K_Q(x,t), 
\]
again, in distributional sense on the interior of $\Omega$.

Recall $K_Q$ is continuous on $\Omega$ and $q_0$ is smooth on $\R$. Thus, if $Q$ is continuous on $\R^+$ ($q$ is too), from the fundamental theorem of calculus $K_Q$ satisfies \eqref{Eq:Wave}.


Assume moreover that $Q$ is in $C^k(\R^{+})$ and let $\ell \in \N$.
From \eqref{Eq:derivative1}, \eqref{Eq:derivative2} and~\eqref{Eq:derivative3}, we infer that ${\mathcal G}_Q^n$ maps $C^\ell(\Omega)$ to $C^{\min(\ell+n,k+1)}(\Omega)$.
It follows then that $K_{Q,n}$ is in $C^{k+1}(\Omega)$ for any $n\in\N$ 
and $K_Q$ is in $C^{k+1}(\Omega)$.

To conclude the proof, let us remark that if $Q$ is $C_{pw}^{k+1}(\R^+)$ (piecewise $C^{k+1}$) then \eqref{Eq:derivative1} and \eqref{Eq:derivative2} provide that if $L\in C^{\ell}(\Omega)$ then $\partial_{\frac{t+x}{2}} \mathcal{G}_{Q,1}(L)$ and $\partial_{\frac{t-x}{2}} \mathcal{G}_{Q,1}(L)$ are $C^{\min(\ell,k+1)}$ in $\frac{t+x}{2}$ and $\frac{t-x}{2}$ respectively. If (moreover) $Q$ is $C^k(\R^+)$
then the right hand side of \eqref{Eq:derivative3} 
is actually $C^{\min(\ell,k)}(\Omega)$ so that $\mathcal{G}_{Q,1}(L)$ is $C^{\min(\ell+1,k+2)}(\Omega)$. 
For $L=K_Q$ this provides $K_Q-K_{Q,1}$ is $C^{k+2}(\Omega)$
as announced.
Since the supports of $K_Q$ and $K_Q-K_{Q,1}$ are contained in $\Omega_0$, we deduce that for $x\in [0,2a]$, 
$$\partial_x^{k+2}K_Q(x^-,(2a-x)^-)=-\frac14\partial_x^{k+1} Q(a^-)$$
which concludes the proof.
\end{proof}



We conclude this paragraph with a continuity result.
\begin{lemm}
Let $a>0$ and $L_a^1$ be the space of integrable functions with support on $[0,a]$ endowed with the $L^1$-norm. The map 
\[
 Q \in L_a^1 \mapsto K_Q \in E_a,
\]
is continuous.
\end{lemm}
\begin{proof}
We consider another potential $\tilde{Q}$ integrable with support in $[0,a]$ and $\tilde{q}=q_0+\tilde{Q}$ such that
\[
M_{\tilde{q},q_0}:=\sup\left\{ \int_0^\frac{t-x}{2}|\tilde{q}(\alpha-\beta)-q_0(\alpha+\beta)|\dd \beta,\; \frac{t-x}{2}\leq \alpha\leq a,\, 0\leq x\leq t\right\}<\infty.
\]
Let $K_{\tilde{Q}}\in E_a$ be the fixed point of ${\mathcal F}_{\tilde Q}$. Then
\begin{align*}
\left(K_{\tilde{Q}}-K_Q\right)(x,t)&=
 \frac{1}{2}\int_\frac{t+x}{2}^\infty \left(\tilde{Q}-Q\right)(s)\dd s\\
& +\int_\frac{t+x}{2}^\infty \int_0^\frac{t-x}{2}(\tilde{q}(\alpha-\beta)-q_0(\alpha+\beta))\left(K_{\tilde{Q}}-K_Q\right)(\alpha-\beta,\alpha+\beta)\dd \beta\dd \alpha\\
&+\int_\frac{t+x}{2}^\infty \int_0^\frac{t-x}{2}(\tilde{Q}(\alpha-\beta)-Q(\alpha-\beta))K_Q(\alpha-\beta,\alpha+\beta)\dd \beta\dd \alpha\\
&=J_{\tilde{Q}-Q,1}(x,t)+{\mathcal G}_{\tilde{Q},1}(K_{\tilde{Q}}-K_Q)(x,t).
\end{align*}
In the above:
$$J_{\tilde{Q}-Q,1}:=K_{\tilde{Q}-Q,1}+\mathcal{H}_{\tilde{Q}-Q}(K_Q),$$
with:
$$\mathcal{H}_{\tilde{Q}-Q}(L)(x,t):=\int_\frac{t+x}{2}^\infty \int_0^\frac{t-x}{2}(\tilde{Q}(\alpha-\beta)-Q(\alpha-\beta))L(\alpha-\beta,\alpha+\beta)\dd \beta\dd \alpha.$$
Let, for $n\in\N$,  
\[
 J_{\tilde{Q}-Q,n+1}={\mathcal G}_{\tilde{Q},1}(J_{\tilde{Q}-Q,n}),
\]
then:
\[
 K_{\tilde{Q}}-K_Q=\sum_{k=1}^n J_{\tilde{Q}-Q,k}+{\mathcal G}_{\tilde{Q},n}(K_{\tilde{Q}}-K_Q).
\]
Recall now~\eqref{Eq:contraction} to deduce that for $L\in E_a$:
\[
 \|\mathcal{G}_{Q,n}(L)\|_\infty\leq  \frac{M_{q,q_0}^na^n}{n!}\|L\|_\infty.
\]
Similarly, for $L\in E_a$, we have:
\[
 \|\mathcal{H}_{\tilde{Q}-Q}(L)\|_\infty\leq m_{\tilde{Q}-Q}a\|L\|_\infty,
\]
where: 
\[
m_{\tilde{Q}-Q}:=\sup\left\{ \int_0^\frac{t-x}{2}|\tilde{Q}(\alpha-\beta)-Q(\alpha-\beta)|\dd \beta,\; \frac{t-x}{2}\leq \alpha\leq a,\, 0\leq x\leq t\right\}<\infty.
\]
Therefore, we obtain:
$$\|J_{\tilde{Q}-Q,1}\|_\infty\leq \frac{1}{2}\|\tilde{Q}-Q\|_1
+
m_{\tilde{Q}-Q} a \|K_Q\|_\infty,$$
and:
$$\|J_{\tilde{Q}-Q,n}\|_\infty\leq \left(\frac{1}{2}\|\tilde{Q}-Q\|_1
+
m_{\tilde{Q}-Q} a \|K_Q\|_\infty\right)  \frac{M_{\tilde{q},q_0}^{n-1}a^{n-1}}{(n-1)!}.$$
Overall we obtain:
\begin{align*}
\|K_{\tilde{Q}}-K_Q\|_\infty
&\leq \sum_{k=1}^n\|J_{\tilde{Q}-Q,k}\|_\infty+\|{\mathcal G}_{\tilde{Q},n}(K_{\tilde{Q}}-K_Q)\|_\infty\\
&\leq \left(\sum_{k=1}^n  \frac{M_{q,q_0}^ka^k}{k!}\right)\left( \frac{1}{2}\|\tilde{Q}-Q\|_1
+
m_{\tilde{Q}-Q} a \|K_Q\|_\infty\right)\\
& +
\frac{(M_{q,q_0}+m_{\tilde{Q}-Q})^na^n}{n!}\|K_{\tilde{Q}}-K_Q\|_\infty
\end{align*}
and hence, for $m_{\tilde{Q}-Q}$ bounded and $n$ sufficiently large, there exists $C>0$ such that:
\begin{align*}
\|K_{\tilde{Q}}-K_Q\|_\infty
&\leq C\left( \frac{1}{2}\|\tilde{Q}-Q\|_1
+
m_{\tilde{Q}-Q} a \|K_Q\|_\infty\right),\\
&\leq C\left( \frac{1}{2}
+
a \|K_Q\|_\infty\right)\|\tilde{Q}-Q\|_1,
\end{align*}
and so $K_Q$ is continuous in $L^\infty$-norm with respect to $Q$ in $L^1$-norm.
\end{proof}

\subsection{Existence and uniqueness} \label{sec:exist_unique}

In this section, we prove that there exists a unique solution $u$ of the Dirichlet problem \eqref{eq:dirichlet_problem} and show that the Dirichlet-to-Neumann operator is well-defined. Due to the possibility of a conical singularity at $r=0$, we cannot appeal to the well-known existence and uniqueness theorems for regular Riemannian manifolds, however the basic results we require follow through without essential modification. 

For clarity, we define $H^{\frac{1}{2}}(\partial M)$ as the completion of $C^{\infty}(\partial M)$ for the following norm based on the spectral data of the Laplacian on $K$ by: 
\[ \lVert \phi \rVert_{H^{\frac{1}{2}}(\partial M)}^2= \sum_{k} (1+\mu_k^2)^\frac{1}{2}|\phi_k|^2. \]
We mean by $H^1(M)$, the closure of $ \mathcal{D}= \{ u\in C^\infty(\overline{M}), ||u||_{H^1(M)} < +\infty\}$\footnote{It should be understood here that these are smooth functions such that all successive derivatives extend continuously to the boundary.} for the natural $H^1(M)$ norm, which we recall is:
\[\begin{aligned} ||u||^2_{H^1(M)}&= \int_M \left(|u|^2 + g(\nabla u, \overline{\nabla u})\right) \dd \textrm{Vol}_g, \\&=\int_{\mathbb{R}_+}\int_K (|u|^2+\frac{1}{f^2}\left(|\partial_x u|^2 + |\nabla_K u|^2 \right) f^n\dd x \dd K\end{aligned} \]

We begin by showing the following lemma:
\begin{lemm}
The trace operator $\gamma : H^1(M) \rightarrow H^\frac{1}{2}(\partial M)$ is a well-defined bounded operator.
\end{lemm}
\begin{proof}
Using that: $|\nabla_K Y_k|^2=(-\Delta_K Y_k, Y_k)= \mu_k^2$, the norm in $H^1(M)$ can be expressed in terms of the decomposition $u=\sum_{k\in\mathbb{N}} u_kY_k$ as:
\begin{equation} \label{eq:norm_decomp} \begin{gathered} ||u||^2_{H^1(M)} = \sum_{k=0}^{+\infty} a_k(u),\\a_k(u)=||u_k||^2_{L^2(\R_+,f^n\dd x)} + \mu_k^2 ||f^{-1}u_k||^2_{L^2(\R_+,f^n\dd x)} + ||f^{-1}u'_k||^2_{L^2(\R_+,f^n\dd x)}.  \end{gathered} \end{equation}
Working on the dense subset $\mathcal{D}$ we have:
\[ |f^{\frac{n}{2}-1}(0)u_k(0)|^2 = -(n-2)\int_0^{+\infty} \frac{f'}{f} |f^{-1}u_k|^2f^{n}\dd t -2\Re \int_0^{+\infty}f^{-1}u'_k f^{-1}\overline{u}_k f^{n}\dd t.\]
Multiply by $(1+\mu_k^2)^{\frac{1}{2}}$ and sum over $k$, then estimate the first term by:
\[(n-2)\left\lVert\frac{f'}{f}\right\rVert_\infty \sum_{k\geq 0 } (1+\mu_k^2) \lVert f^{-1}u_k\rVert^2_{L^2(\mathbb{R}_+, f^{n}\dd x)}, \]
and the second term, using the Cauchy-Schwarz inequality, by:
\[2\left(\sum_{k\geq 0} \lVert f^{-1}u'_k \rVert^2_{L^2(\R_+,f^{n}\dd x)} \right)^{\frac{1}{2}}\left( \sum_{k\geq 0} (1+\mu_k^2)||f^{-1}u_k||^2_{L^2(\R_+, f^{n}\dd x)} \right)^{\frac{1}{2}}, \]
thus:
\[ \begin{aligned} \sum_{k\geq 0}(1+\mu_k^2)^{\frac{1}{2}}f^{n-2}(0) \left\lvert u_k(0)\right\rvert^2 \leq C ||u||^2_{H^1(M)}.  \end{aligned}\]
with, $C=\max((n-2)\left\lVert\frac{f'}{f}\right\rVert_\infty\lVert f^{-1}\rVert_\infty^2,(n-2)\left\lVert\frac{f'}{f}\right\rVert_\infty, ||f^{-1}||^2_\infty)$.
We conclude by a standard density argument.
\end{proof}

We shall now prove the existence and uniqueness of the unique $u\in H^1(M)$ solution of~\eqref{eq:dirichlet_problem}. In fact, we shall instead work with variable $v=f^{\frac{n}{2}-1}u$ and solve~\eqref{eq:dirichlet_problem2}. We begin with an important lemma:
\begin{lemm}
\label{lemm:hs}
Let $n\geq 3$, $\lambda\not\in\sigma_{pp}(-\Delta_g)\footnote{It should be understood here that $-\Delta_g$ is viewed as an operator on $L^2(M)$ with homogenous Dirichlet boundary condition.}$, $\phi \in H^{\frac{1}{2}}(M)$. If $v$ is a solution of~\eqref{eq:dirichlet_problem2} such that $u=f^{-\frac{n}{2}+1}v \in H^1(M)$ then the components $v_k$ are uniquely determined and satisfy $v_k,v_k' \in L^2(\R_+, \textrm{d}x)$. In particular, $v$ is unique.
\end{lemm} 
\begin{proof}
Let us first begin by expressing $||u||_{H^1(M)}$ in terms of $v$ and its components $v_k$. We have (see Equation~\eqref{eq:norm_decomp}):
\[\begin{aligned} a_k(u)&= ||v_k||^2_{L^2(\mathbb{R}_+,f^2\dd x)} + \mu_k^2 ||v_k||^2_{L^2(\mathbb{R}_+,\dd x)} + || (1-\frac{n}{2})\frac{f'}{f}v_k + v_k' ||^2_{L^2(\mathbb{R}_+,\dd x)},\\&\equiv ||v_k||_k. \end{aligned}\]
Decomposing onto harmonics, we see that $v_k$ must satisfy~\eqref{DirichletProblem1D} and lie in the Hilbert space $\mathcal{H}_k$ with norm $||\cdot||_k$ defined by the above expression. For $k\geq 1$, $\mu_k \neq 0$ so that $\mathcal{H}_k \hookrightarrow L^2(\mathbb{R}_+, \textrm{d}x)$, continuously. Standard theory of ODE's then imply that $v_k$ exists and is uniquely determined thus proving Lemma~\ref{lemm:hs} for these harmonics.  

The case $k=0$ requires further analysis of the behaviour of solutions for large values of $x$, existence and unicity being guaranteed by standard ODE results. Furthermore, according to our assumptions~(\ref{hyp:cf'},\ref{hyp:cf2}), $x\mapsto f(x)-e^{-x}$ is compactly supported so working away from its support we can assume: $f(x)=e^{-x}$. We are therefore interested in the asymptotic behaviour of solutions of the equation:
\[-v_0'' -\lambda e^{-2x}v_0=-z_0^2v_0,\]
where $z_0^2=\frac{(n-2)^2}{4}$.

Setting $x=-\ln \frac{t}{\sqrt{\lambda}}$, the function $y$ defined by $y(t)=v_0(x(t))$ satisfies the Bessel equation:
\[ t^2y''(t) + ty'(t)+(t^2-z_0^2)y(t)=0. \]
According to the parity of the integer $n\geq 2$, $z_0$ is either a positive integer or half-integer:
\begin{itemize}
\item If $z_0$ is a half-integer $\geq \frac{3}{2}$ then the generic solution is a linear combination of: 
\[\begin{aligned} v_{\pm}(x)&=J_{\pm z_0}(\sqrt{\lambda}e^{-x})\\&=\left(\frac{\sqrt{\lambda}}{2}\right)^{\pm z_0}e^{\mp z_0 x}\sum_{m\geq 0}\frac{(-1)^m}{m!\Gamma(z_0+m+1)}\left(\frac{\sqrt{\lambda}e^{-x}}{2}\right)^{2m}. \end{aligned}\]
At infinity, one has:
\[ v_{\pm}(x)^2e^{-2x} \underset{x\to +\infty}{\sim}\left(\frac{\sqrt{\lambda}}{2}\right)^{\pm 2z_0}e^{2(\mp z_0-1)x},\]
from which it follows that all solutions that are in $L^2(\mathbb{R}_+,e^{-2x}\dd x)$ are constant multiples of $v_+$ and therefore also belong to $L^2(\mathbb{R},\textrm{d}x)$.

\item In the special case $z_0=\frac{1}{2}$ ($n=3$), both solutions are in $L^2(\mathbb{R},e^{-2x}\textrm{d}x)$. However, we also require $v'+v \in L^2(\R_+,\textrm{d}x)$. Now: 
\[ J_{-\frac{1}{2}}(t)=\sqrt{\frac{2}{\pi t}}\cos(t), \quad J_{\frac{1}{2}}(t)=\sqrt{\frac{2}{\pi t}}\sin(t), \]
and:
\[ \begin{aligned} v'_+(x)+v_+(x) &=\sqrt{\frac{2}{\pi\lambda}}\left(\frac{3}{2} e^{\frac{1}{2}x}\sin(\sqrt{\lambda}e^{-x})-e^{-\frac{1}{2}x}\sqrt{\lambda}\cos(\sqrt{\lambda}e^{-x}) \right), \\ v'_-(x)+v_-(x)&=\sqrt{\frac{2}{\pi\lambda}}\left(\frac{3}{2}e^{\frac{1}{2}x}\cos(\sqrt{\lambda}e^{-x}) + e^{-\frac{1}{2}x}\sqrt{\lambda}\sin(\sqrt{\lambda}e^{-x})\right), \end{aligned}  \]
therefore only $v_+$ satisfies the additional condition. So all solutions that satisfy the required condition are constant multiples of $v_+$ and are additionally elements of $L^2(\R_+, \dd x)$.

\item When $z_0$ is a non-zero positive integer, then the solution $v_+(x)=J_{z_0}(\sqrt{\lambda}e^{-x})$ is clearly in $L^2(\mathbb{R}_+,\dd x) \hookrightarrow L^2(\mathbb{R}_+,e^{-2x}\textrm{d}x)$ furthermore $v_+' \in L^2(\mathbb{R}_+,\textrm{d}x)$. Another linearly independent solution is given by means of Bessel functions of the second kind:
\[  \begin{aligned} v_-(x)= e^{z_0x} \left(\frac{\sqrt{\lambda}}{2} \right)^{-z_0}\sum_{r=0}^{z_0-1}\frac{\Gamma(z_0-r)}{\Gamma(r+1)}\left(\frac{\sqrt{\lambda}e^{-x}}{2}\right)^{2r} +2xJ_{z_0}(\sqrt{\lambda}e^{-x})\\+\sum_{r=0}^\infty \frac{(-1)^r}{\Gamma(r+1)\Gamma(z_0+r+1)}(\psi(r)+\psi(n+r))\left(\frac{\sqrt{\lambda}e^{-x}}{2} \right)^{z_0+2r}, \end{aligned} \]
with $\psi(r)= [\frac{\dd}{\dd t}\ln \Gamma(t+1)]_{t=r}$. Hence all solutions that are in $L^2(\R_+, e^{-2x}\dd x)$ are constant multiples of $v_+$.

\item Finally, in the special case $z_0=0$, $(n=2)$, the solution $v_+$ satisfies $v_+ \in L^2(\mathbb{R}_+,e^{-2x}\dd x)$ and $v'_+ \in L^2(\R_+,\dd x)$. A linearly independent solution is given by:
\[ v_-(x)= -xJ_0(\sqrt{\lambda}e^{-x}) + \sum_{r=1}^{+\infty}\frac{(-1)^{r-1}}{\Gamma(r+1)^2}\left(\frac{\sqrt{\lambda}e^{-x}}{2}\right)^{2r}\psi(r). \]
Whilst it satisfies $L^2(\mathbb{R}_+,e^{-2x}\dd x)$ it does not satisfy $v'\in L^2(\R_+, \dd x)$, hence all solutions that satisfy the required condition are constant multiples of $v_+$. Note however that in this very specific case $v_+ \not \in L^2(\R_+, \dd x)$.
\end{itemize}
\end{proof}

We conclude this section with the existence result:
\begin{lemm}
Let $n\geq 3$, $\lambda \not\in \sigma_{pp}(-\Delta_g)$, $\phi \in H^{\frac{1}{2}}(M)$ and $(v_k)_{k\in \mathbb{N}}$ be the functions given by Lemma~\ref{lemm:hs}, then $u=\sum_k f^{1-\frac{n}{2}}v_k Y_k$ is an element of $H^1(M)$ and is therefore the unique solution of~\eqref{eq:dirichlet_problem}.
\end{lemm}
\begin{proof}
We must show that $\sum_{k} ||v_k||_k^2 < +\infty$, for which it will be sufficient to show that there are positive constants $C_1,C_2>0$ such that for large enough $k \in \N$, 
\begin{align}
||v_k||^2_{L^2(\mathbb{R},\dd x)} \leq \frac{C_1}{z_k}|\phi_k|^2, \\
||v_k'||^2_{L^2(\mathbb{R},\dd x)} \leq C_2 z_k |\phi_k|^2.
\end{align}
Let $K$ be the transformation operator of Section~\ref{sec:transfo} and define:
\[\psi(x,z)= J_z(\sqrt{\lambda}e^{-x}) + \int_{x}^{+\infty} K(x,t)J_z(\sqrt{\lambda}e^{-t})\dd t, \]
using the properties of $K$, one sees that:
\begin{equation} v_k(x) = \frac{\psi(x,z_k)}{\psi(0,z_k)}f^{\frac{n}{2}-1}(0)\phi_k, \quad z_k=\sqrt{\mu_k^2 +\left(\frac{n}{2}-1\right)^2}. \end{equation}
We begin with a rough estimate of $|\psi(x,z_k)|$. Note that the condition $\lambda \not\in \sigma_{pp}(-\Delta_g)$ guarantees that it does not vanish for any $k$. Using Lemma~\ref{lemm:asymptotic_bessel}, for large enough $k$, one has uniformly for $t\in \R_+$ (since the Bessel functions $t\mapsto J_{z_k}$ are regular at $0$),
\[ |J_{z_k}(\sqrt{\lambda}e^{-t})| \leq 2 \frac{\lambda^{\frac{z_k}{2}}e^{-tz_k}}{2^{z_k}\Gamma(z_k+1)}. \]
Therefore, for large enough $k$:
\[|\psi(x,z_k)| \leq 2 \frac{\lambda^{\frac{z_k}{2}}}{2^{z_k}\Gamma(z_k+1)}\left(1 + \frac{||K||_\infty}{z_k} \right) e^{-xz_k}. \]
Hence, for large enough $k$:
\[\int_{\R_+} |\psi(x,z_k)|^2\dd x \leq  \frac{2}{z_k}\left(\frac{\lambda^{\frac{z_k}{2}}}{2^{z_k}\Gamma(z_k+1)}\left(1 + \frac{||K||_\infty}{z_k} \right)\right)^2.\]

We also require a lower bound on $\psi(0,z)$, for this we use (in advance) Lemma~\ref{eq:jost_re_pos} that enable us to show that, for large enough $k$, one has~:
\begin{equation} \label{jost_lower_bound_rhp} |\psi(0,z_k)| \geq \frac{1}{2} \frac{\lambda^\frac{z_k}{2}}{2^{z_k}\Gamma(z_k+1)}. \end{equation}
Thus:
\[ ||v_k||^2_{L^2(\R_+,\dd x)} \leq \frac{8}{z_k}\left(1+\frac{||K||_\infty}{z_k}\right)^2|f^{\frac{n}{2}-1}(0)|^2|\phi_k|^2.\]
This proves the first estimate.
For the second, write:
\[\psi'(x,z_k)= \sqrt{\lambda}e^{-x}J_{z_k}'(\sqrt{\lambda}e^{-x}) -K(x,x)J_{z_k}(\sqrt{\lambda}e^{-x}) + \int_x^{+\infty}\partial_xK(x,t)J_{z_k}(\sqrt{\lambda}e^{-t})\dd t.\]
Then, using the recurrence relation for Bessel functions, one has:
\[ J_{z_k}'(\sqrt{\lambda}e^{-x})=\frac{z_k}{\sqrt{\lambda}e^{-x}}J_{z_k}(\sqrt{\lambda}e^{-x}) - J_{z_k+1}(\sqrt{\lambda}e^{-x}). \]
So that first term can be estimated for $k$ large enough by:
\[ \begin{aligned} \left\lvert z_k J_{z_k}(\sqrt{\lambda}e^{-x})-\sqrt{\lambda}e^{-x}J_{z_k+1}(\sqrt{\lambda}e^{-x})\right\rvert &\leq2 \frac{\lambda^{\frac{z_k}{2}}e^{-z_kx}}{2^{z_k}\Gamma(z_k+1)}(z_k + \frac{\lambda e^{-2x} }{2z_k}),\\& \leq  4z_k\frac{\lambda^{\frac{z_k}{2}}e^{-z_kx}}{2^{z_k}\Gamma(z_k+1)}.  \end{aligned}\]
The second and third terms can respectively be estimated for large $k$ by:
\[\begin{gathered} |K(x,x)J_z(\sqrt{\lambda }e^{-x})| \leq 2 \frac{\lambda^\frac{z_k}{2}e^{-z_kx}}{2^{z_k}\Gamma(z_k+1)}||K|| _\infty, \\ \left\lvert \int_x^{+\infty}\partial_xK(x,t)J_{z_k}(\sqrt{\lambda}e^{-t})\dd t \right\rvert \leq  2 \frac{\lambda^\frac{z_k}{2}e^{-z_kx}}{2^{z_k}\Gamma(z_k+1)}\frac{||\partial_x K||_{\infty}}{z_k}.\end{gathered}\]
Hence for large enough $k$:
\[ |\psi'(x,z_k)|\leq \frac{8z_k\lambda^{\frac{z_k}{2}e^{-z_kx}}}{2^{z_k}\Gamma(z_k+1)}. \]
Applying again~\eqref{jost_lower_bound_rhp}, we now conclude that, for large enough $k$:
\[ ||v'||^2_{L^2(\R_+,\dd x)}\leq 16 z_k f^{n-2}(0)|\phi_k|^2.  \]
\end{proof}

\section{Jost functions and Regge poles}
We are now ready to address the main objects of our study: the Regge poles. Recall from Definition~\ref{defi:regge_poles} that they are the poles of the meromorphic extension of the Weyl-Titchmarsh function of the Schrödinger operator: 
\[H= -\frac{d^2}{dx^2} + Q_f(x) - \lambda e^{-2x},\]
with Dirichlet boundary conditions at $x=0$.

We shall first discuss the reference case in which $Q_f=0$, and collect information about the corresponding \emph{Jost functions} of which zeros are exactly what we refer to as the Regge poles. 
We recall the following terminology from the context of $1$-dimensional Schrödinger operators with a potential $q\in L^1_{\textrm{loc}}(\R_+)$.
Let $z \in \mathbb{C}, \Re z >0$.
\begin{itemize}
\item The \emph{Jost solution}, $\psi(\cdot, z)$ is the unique solution of the Schrödinger equation:
\[ -\psi''(x, z) + q(x)\psi(x,z)=-z^2\psi(x,z),\]
that satisfies the condition: $\psi(x,z)\underset{x\to\infty}\sim e^{-xz}.$ Note that the notation $'$ will always indicate derivatives with respect to $x$.
\item The \emph{Jost function} is defined to be $z\mapsto \psi(0,z)$.
\end{itemize}

It is well-known that (see for instance \cite{Si1999, Be2001}), the Weyl-Titchmarsh function is then:
\[ M(-z^2)= \frac{\psi'(0,z)}{\psi(0,z)}.\]
Thus, if we can extend $z\mapsto \psi(0,z)$ to an entire function on the whole complex plane, $M$ extends meromorphically to the whole Riemann surface: $\{(w,z)\in \mathbb{C}^2, w+z^2=0\}\cong \mathbb{C}$, and the poles of this extension correspond to zeros of $\psi$. For clarity, and since the Riemann surface is holomorphic to $\mathbb{C}$, we shall view this extension as a meromorphic function on $\mathbb{C}$, that we denote by: $z\mapsto m(z)$.

\subsection{The reference case $Q_f=0$} \label{Unperturbed}
In this section, we assume that $Q_f=0$.  Referring now to Section~\ref{sec:exist_unique}, we know that in this case, the Jost solution as defined above is given by a rescaled Bessel function:
\begin{equation}\tilde{\psi}(x,z)=\Gamma(z+1)2^z\lambda^{-\frac{z}{2}}J_z(\sqrt{\lambda}e^{-x}). \end{equation}
We can observe that this normalisation introduces poles in the Jost solution. However, the Weyl-Titschmarsh function, which is the object of true intrinsic interest to our problem, as a quotient, will be insensitive to an overall normalisation factor that depends only on $z$. We will therefore work with what we shall call the \emph{modified Jost solution}, which is a rescaling of these quantities that removes these poles. In the case $Q_f=0$ we shall simply set:
\[ \psi_0(x,z)=J_z(\sqrt{\lambda}e^{-x}). \]

We can now apply the results of Section~\ref{sec:bessel}, that show that for any fixed $x\in \R_+$, $\psi_0(x,\cdot)$ analytically extends to an entire function of order $1$ and infinite type. Let us denote by $m_0$ the meromorphic extension of the Weyl-Titchmarsh function in this case.

We can obtain some basic information about the distribution of the Regge poles in this case, as zeros of $\psi_0(0,z)$. The asymptotics given by Lemma~\ref{lemm:asymptotic_bessel} in Section~\ref{sec:bessel} is a key ingredient to our results.  We first have:
\begin{lemm}
There are an infinite number of Regge poles when $Q_f=0$. Asymptotically, they are simples poles on the negative real line approaching the negative integers.
\end{lemm}
\begin{proof}
Appealing to Equation~\eqref{eq:fa}, for $|z|$ sufficiently large in $U_\delta$, we have:
\[ \left\lvert J_z(\sqrt{\lambda}) - \frac{\left(\frac{\sqrt{\lambda}}{2}\right)^z}{\Gamma(z+1)}\right\rvert < \left\lvert\frac{\left(\frac{\sqrt{\lambda}}{2}\right)^z}{\Gamma(z+1)}\right\rvert.\]
Rouché's lemma \cite[Theorem 10.36]{rudin1987real} then shows that the zeros of $z\mapsto J_z(\sqrt{\lambda})$ are same as those of $z\mapsto \frac{\left(\frac{\sqrt{\lambda}}{2}\right)^z}{\Gamma(z+1)}$, on any disk completely contained in the intersection of $U_\delta$ and the outside of a large enough disk centered in $0$. Since the second function has no zeros in this region, neither does $J_z(\sqrt{\lambda}).$ This proves the second point. 

Turning this around if $\delta' > \delta$ is such that, for any $n \in \mathbb{Z}^*_-$ $S(-n,\delta') \subset U_\delta$, then Rouché's lemma shows that for large enough $|n|$, there is exactly one zero in such a disk. The zeros are in fact symmetric with respect to the real axis, so it must be on the negative real axis.
\end{proof}

Lemma~\ref{lemm:asymptotic_bessel} also gives directly the asymptotic behaviour of $m_0$ on $U_\delta$:
\begin{lemm}
\label{lemm:wt_bessel_asymp}
For $z \in U_\delta$:
\[\left\lvert m_0(z) \right\rvert = \left\lvert \frac{J'_z(\sqrt{\lambda})}{J_z(\sqrt{\lambda})}\right\rvert = \gO{|z|}{|z|}{\infty} \]
\end{lemm}
\begin{proof}
Let $t>0$, starting from the recurrence relation:
\[ J_z'(t)=\frac{z}{t}J_z(t) - J_{z+1}(t),\] one has, for large enough $|z|$: \[\frac{J_z'(t)}{J_z(t)}= \frac{z}{t}- \frac{J_{z+1}(t)}{J_z(t)}.\]
The result then follows from the asymptotics given in Lemma~\ref{lemm:asymptotic_bessel}, that show that: \[\left\lvert\frac{J_{z+1}(t)}{J_z(t)}\right\rvert= \gO{\frac{1}{|z|}}{|z|}{\infty}.\qedhere\]
\end{proof}

\subsection{The perturbed case $Q_f \neq 0$} \label{Perturbed}

By definition of the transformation operator in Section~\ref{sec:transfo}, the \emph{modified} Jost solution in the general case is given by: 
\begin{equation} \label{eq:modified_jost} \psi(x,z)= J_z\left(\sqrt{\lambda}e^{-x}\right) + \int_{x}^{+\infty} K(x,s)J_z\left(\sqrt{\lambda}e^{-s}\right)\dd s. \end{equation}

Using the properties of the transformation operator, we can now generalise the lemmata of the previous section:
\begin{lemm}
\label{lemm:basic_jost}
Let $x\in \R_+$, then:
\begin{enumerate} \item $z\mapsto \psi(x,z)$ is an entire function,
 \item there are constants $c_1,c_2 >0$ such that: 
\[|\psi(x,z)| \leq c_1e^{c_2|z|\ln(|z|)},\]
for $|z|$ large enough.
\end{enumerate}
\end{lemm}
\begin{proof} 
The key point is that, since the potential is of compact support, $K(x,\cdot)$ is equally of compact support. The first statement then follows from an easy refinement of the estimates in the proof of Lemma~\ref{lemm:ordre_bessel} to show that for any compact subset $[\alpha, \beta]\subset \R_+^*, R\in \R_+^*$ one can find $M\in \R_+^*$, \[\forall |z|\geq R, t \in [\alpha,\beta],|J_z(t)|\leq M.\] The second statement is also an immediate consequence of the estimate in Lemma~\ref{lemm:ordre_bessel} and continuity of the kernel $K$. 
\end{proof}

\subsubsection{Asymptotics of $z\mapsto \psi(0,z)$}
We shall now use the asymptotics of Section~\ref{lemm:ordre_bessel} to derive asymptotics for $\psi$, in particular in the half plane $\Re z <0$.

To begin our development let us first reformulate the asymptotics of Lemma~\ref{lemm:asymptotic_bessel} as follows. For $z \in U_\delta$:
\[J_z(\sqrt{\lambda}e^{-s})=\frac{\lambda^{\frac{z}{2}}e^{-sz}}{2^z\Gamma(z+1)}\left(1+ R(s,z) \right), \]
where, $R(s,z)= \tilde{R}(\sqrt{\lambda}e^{-s},z)$ with: 
\[\tilde{R}(t,z)=\sum_{m=1}^{+\infty} \frac{(-1)^m}{m!(z+1)_{(m)}}\left(\frac{t}{2}\right)^{2m},  \]
and $\tilde{R}(t,z)=\gO{\frac{1}{|z|}}{|z|}{\infty}$ uniformly for $t$ in a compact set.
Injecting this into equation~\eqref{eq:modified_jost} evaluated at $x=0$, we have:
\[\begin{aligned}\psi(0,z)= \frac{\lambda^{\frac{z}{2}}}{2^z\Gamma(z+1)}\left(1 + R(0,z) +I_1(z) + I_2(z)\right),\end{aligned} \]
with:
\[\begin{gathered}I_1(z)=\int_0^{2a}K(0,s)e^{-sz}\dd s, \\ I_2(z) = \int_{0}^{2a}K(0,s)R(s,z)e^{-sz}\dd s. \end{gathered} \]

In the half-plane $\mathbb{C}_+=\{\Re z \geq 0\}$, the Riemann-Lebesgue lemma shows, without any further assumption on the potential $Q_f$, that $I_1(z)$ vanishes in the limit $|z| \to \infty$ and $z\in \C_+$. Using the estimate on $R$, we see that $I_2(z)=\gO{\frac{1}{|z|}}{|z|}{\infty}$ as long as $z \in \C_+$. Hence, in summary:
\begin{lemm}
\label{lemm:asymp_pos_hp}
\begin{equation}\label{eq:jost_re_pos} \psi(0,z)=\displaystyle \frac{\lambda^{\frac{z}{2}}}{2^z\Gamma(z+1)}\left(1+ \po{1}{|z|}{\infty} \right), \quad \Re z \geq 0.\end{equation}
\end{lemm}
Lemma~\ref{lemm:asymp_pos_hp} determines some universal asymptotics for Jost solutions within the class of potentials we consider. Combined with Lemma~\ref{lemm:basic_jost}, this can be used to assert that:
\begin{prop}
\label{prop:jost_unique}
If two warped balls have the same Regge poles then the modified Jost function's of the associated Schrödinger problem are the same.
\end{prop}
\begin{proof}
The Jost function is entire and of finite order $\rho=1$, hence by the Hadamard factorisation theorem~\cite[\S 4.2]{Le1996}:
\[ \psi(0,z)=z^k \exp(\alpha + z \beta)\prod_{j=1}^\infty \left(\left(1-\frac{z}{z_j}\right)e^{\frac{z}{z_j}}\right)^{n_j}.\]
Where, $\{z_j\}$ are the non-zero Regge poles, $\{n_j\}$ their multiplicity and $k \in \{0,1\}$ according to whether zero is a pole or not.
It follows then that: 
\[\prod_{j=1}^\infty \left(\left(1-\frac{z}{z_j}\right)e^{\frac{z}{z_j}}\right)^{n_j} \sim \frac{\lambda^{\frac{z}{2}}}{2^z\Gamma(z+1)}z^{-k}e^{-\alpha-z\beta},  \]
in the half plane $\Re z \geq 0$.
Hence, $\alpha,\beta$ are determined by the asymptotic behaviour of the canonical product associated to the $\{z_j\}$ in this same half-plane, and therefore depends only on their distribution.
\end{proof}
\begin{rema}
At least in principle, studying the asymptotic behaviour of the canonical product should enable to determine $\alpha$ and $\beta$.
\end{rema}

In order to have precise asymptotics in the half-plane $\Re z <0$, where the dominant terms arise from the Laplace type integrals $I_1,I_2$, we require further assumptions on the potential $Q_f$. As mentioned in the introduction,  we make the assumption that there is a jump in  the $(p-1)$th derivative of $Q_f$, $p\geq 1$, at the boundary point $a$ of the support of the potential. As we observed in Section~\ref{sec:transfo}, this leads to a jump in the $p$th derivative of the kernel $K$. This allows us to handle the Laplace integrals $I_1, I_2$ and extract a dominant term by integration by parts. Indeed, after $p+1$ integrations by parts:
\[\begin{aligned} I_1(z)=&\sum_{k=1}^{p+1}\frac{1}{z^k}(\partial^{k-1}_sK)(0,0) -\frac{1}{z^{p+1}}\underbrace{\partial^p_sK(0,2a^{-})}_{\neq 0}e^{-2az}\\&\hspace{2in}+\frac{1}{z^{p+1}}\int_0^{2a}(\partial^{p+1}_sK)(0,s)e^{-sz}\dd s. \end{aligned}\]
The dominant term is clearly: $\displaystyle \frac{1}{z^{p+1}}\underbrace{\partial^p_sK(0,2a^{-})}_{\neq 0}e^{-2az}$, the other boundary terms from the integration by parts are negligible compared to this in the half-plane $\{\Re z <0\}$. Without further hypothesis, the final term is $\displaystyle \po{\frac{e^{2az}}{z^{p+1}}}{|z|}{\infty}$.
Indeed: \[\int_0^{2a}\partial^{p+1}_s K(0,s)e^{(2a-s)z}\dd s=\int_0^{2a}\partial_s^{p+1}K(0,2a-u)e^{uz}\dd s,\]
and the latter integral vanishes in the limit $|z|\to \infty$ if $\Re z \leq 0$ by the Riemann-Lebesgue lemma. 
\begin{rema}
Note that the estimate of the remainder can be improved to: $\displaystyle \gO{\frac{e^{2a(\Re z)_-}}{|z|^{p+2}}}{|z|}{\infty}$, where $(x)_- = \frac{-x +|x|}{2}$, denotes the negative part\footnote{which is positive.} of $x\in \R$, in the intersection of $U_\delta$ with any sector of the form:
\[ \Re z < -\varepsilon | \Im z |,\]
where $\varepsilon >0$. 
\end{rema}
The same method will enable us to treat the remainder term $I_2(z)$. We will need the following Lemma which can be proved in a similar fashion to Lemma~\ref{lemm:asymptotic_bessel}:
\begin{lemm}
Let $p\geq 1$, then:
\[ \partial^p_s R(s,z)=\lambda^{\frac{p}{2}}\sum_{m=1}^{+\infty} \frac{(-1)^{m+p}m^p}{m!(z+1)_{(m)}}\left(\frac{\sqrt{\lambda}e^{-s}}{2}\right)^{2m}, \]
furthermore for $z\in U_\delta$,
\[ \left\lvert\partial^p_sR(s,z)\right\rvert\leq \frac{\lambda^{\frac{p}{2}}\delta}{|z+1|}\sum_{m=1}^{+\infty}\frac{m^p}{m!}\left( \frac{\lambda e^{-2s}}{4\delta}\right)^m.  \]
Therefore, \[\partial^p_sR(s,z) =\gO{\frac{1}{|z|}}{|z|}{\infty},\]
uniformly for $s \in \R_+.$
\end{lemm}

\noindent Now, as before, we can do $p+1$ integration by parts in $I_2(z)$ to obtain: \[I_2(z)=R_1(z)+R_2(z)+R_3(z),\] with:
\begin{equation} \begin{gathered}R_1(z)= \sum_{k=1}^{p+1}\frac{1}{z^k}\sum_{m=0}^{k-1}(\partial^{m}_sK)(0,0)\partial^{k-1-m}_sR(0,z), \\ R_2(z)=-\frac{1}{z^{p+1}}\sum_{k=0}^p\partial^k_sK(0,2a^{-})\partial^{p-k}_sR(2a,z)e^{-2az}, \\R_3(z)=\frac{1}{z^{p+1}}\sum_{k=0}^{p+1}\int_0^{2a}\partial^{k}_sK(0,s)\partial_s^{p+1-k}R(s,z)e^{-sz}\dd s. \end{gathered} \end{equation}

We estimate each of these terms as follows: \[\begin{gathered} R_1(z)=\gO{\frac{1}{|z|^2}}{|z|}{\infty}, \quad R_3(z)= \gO{\frac{e^{2a(\Re z)_-}}{|z|^{p+2}}}{|z|}{\infty}\\ R_2(z)=-\frac{1}{z^{p+1}}\partial_s^pK(0,2a^{-})R(2a,z)e^{-2az}=\gO{\frac{e^{2a(\Re z)_-}}{|z|^{p+2}}}{|z|}{\infty}.\end{gathered} \]
Hence, we have shown that: 
\begin{lemm} \label{lemm:asymp_saut}
If $Q_f \in C^{p-2}(\R_+)\cap AC^{p-1}(\R_+)$ and its $(p-1)th$ derivative is continuous, save a jump at the boundary point $a$, $p\geq 1$, then in the open half-plane $\Re z <0$:
\begin{equation*} \psi(0,z) =\frac{\lambda^{\frac{z}{2}}}{2^z\Gamma(z+1)}\left( 1+ \gO{\frac{1}{|z|}}{|z|}{\infty} - \partial^{p}_{s}K(0,2a^{-})\frac{e^{-2az}}{z^{p+1}}\left(1+\po{1}{|z|}{\infty} \right) \right).   \end{equation*}
\end{lemm}
\noindent Which can be improved a little in sectors:
\begin{lemm}
\label{lemm:asymp_sector}
Let $\varepsilon>0$ and $0<\delta < \frac{1}{2}$. Under the same assumptions on the potential $Q_f$ as in the previous lemma, we have the following asymptotics:
\[ \psi (0,z) =\frac{\partial^p_sK(0,2a^{-})\lambda^{\frac{z}{2}}e^{-2az}}{2^z\Gamma(z+1)z^{p+2}}\left(1+ \po{1}{|z|}{\infty} \right),\]
in the domain $\{ \Re z <-\varepsilon|\Im z| \} \cap U_\delta$.
\end{lemm}

We are now in a position to prove Theorem \ref{MainAsymptotic}. 

\begin{proof}[Proof of Theorem \ref{MainAsymptotic}]
	
Given $\varepsilon>0$, we consider first the sector $\{ \Re z <-\varepsilon|\Im z| \}$ that contains the negative real axis. Using the asymptotics given in Lemma \ref{lemm:asymp_sector}, it follows easily from Rouché's theorem applied to each disk $D(-k,\delta),  k \in \N,  0<\delta<\frac{1}{2}$ that for large enough $|z|$, the zeros of $\psi(0,z)$ are simple, real and contained in $D(-k,\delta)$. These zeros correspond to the $(\alpha_k)_{k \geq0}$'s in the statement of Theorem \ref{MainAsymptotic}. 

Consider now the sectors $\{\Re(z) < 0\} \cap \{ \Re z > -\varepsilon|\Im z| \}$. From the asymptotics given in Lemma \ref{lemm:asymp_saut}, the zeros of $\psi(0,z)$ coincide there with the zeros of the function:
$$
  \left( 1+ \po{1}{|z|}{\infty} - \partial^{p}_{s}K(0,2a^{-})\frac{e^{-2az}}{z^{p+1}}\left(1+\po{1}{|z|}{\infty} \right) \right). 
$$
However, it follows from the method of Hardy \cite{Ha1905} and Cartwright \cite{Ca1930, Ca1931} that these zeros - denoted by $(\beta_j)_{j \in \Z^*}$ - satisfy the asymptotics:
$$
\begin{aligned}
   \beta_{\pm j} = \pm i \frac{\pi}{2a} \left( 2|j| + \frac{p-1}{2} \pm (sgn(A) + 1)   \right) &- \frac{p+1}{2a} \log \frac{|j| \pi}{a} \\&+ \frac{1}{2a} \log |A| (p-1)! + o(1),
   \end{aligned}
$$
where $A =  \frac{(-1)^p \partial_s^p K(0,2a^-)}{(p-1)!}$. 
This concludes the proof of Theorem \ref{MainAsymptotic}.

\end{proof}

\section{Solving the inverse problem} \label{InversePb}

\subsection{The Gelfand-Levitan-Marchenko method}

In the introduction of~\cite{MSW2010}, the authors outline a general scheme for the proof of the uniqueness result we seek, based on the theory of Gelfand-Levitan-Marchenko~\cite[Chapter IV]{Le2018}. The advantage of this approach requires only a rather mild assumption on the overall potential: $q(x)=-\lambda e^{-2x} +Q_f(x)$; that is: \[\int_0^{+\infty} x|q(x)|\dd x < + \infty. \]
This condition is clearly satisfied by the class of potentials we consider. 

The proof however is based on a different set of spectral data, nevertheless we will show that it can be obtained from the knowledge of the Regge poles. Let us begin by introducing some notation. First, we must revert to the usual normalisation for the Jost solutions, and define:
\begin{equation} \label{eq:mjost_jost}\tilde{\psi}(x,z)=\Gamma(z+1)2^z\lambda^{-\frac{z}{2}}\psi(x,z). \end{equation}
Consider now the $S$-function for $k\in \R$: 
\[ S(k)=\frac{\tilde{\psi}(0,ik)}{\tilde{\psi}(0,-ik)}=\frac{\tilde{\psi}(0,ik)}{\overline{\tilde{\psi}(0,ik)}},\] and set: \[F_S(x)=\frac{1}{2\pi} \int_\R [1-S(k)]e^{ikx}\textrm{d}k.\]

Zeros on the positive half-line of the (modified) Jost function correspond to negative Dirichlet eigenvalues. For nicely decaying potentials like those we consider, they are finite in number and simple (see for example~\cite[\S 4.2, Property 5, p79]{Le2018}). Let us denote them: $0<\alpha_1<\alpha_2<\dots< \alpha _N$. Corresponding eigenfunctions are given by: $\phi_k=\tilde{\psi}(.,\alpha_k).$ Define now:
\[ m_k = \int_0^\infty |\phi_k|^2\dd x =- \frac{\tilde{\psi}'(0,\alpha_k)\dot{\tilde{\psi}}(0,\alpha_k)}{2\alpha_k}, \quad k \in \{1,\dots,k\}.\]
In the above $\dot{}$ denotes differentiation with respect to $z$.
Now set:
\[F(x)= \sum_{k=1}^N\frac{e^{-\alpha_kx}}{m_k}+F_S(x).\]

\begin{lemm}
Suppose that two warped balls have the same Regge poles, then their respective $S$-functions and the coefficients $m_k$ coincide.
\end{lemm}
\begin{proof}
In Proposition~\ref{prop:jost_unique}, we have already established equality of the modified Jost functions, it follows from Equation~\eqref{eq:mjost_jost}, that the usual Jost functions coincide also. It follows immediately that the functions $S$ and $F_S$ must be identical.

In order to show that the $m_k$ must coincide it suffices to show that $\tilde{\psi}'(0,\alpha_k)$ can be determined from the Jost function $\tilde{\psi}(0,\alpha_k)$. However, the Wronskian is known independently of the potential $Q_f$ to satisfy:
\[W(\tilde{\psi}(0,z),\tilde{\psi}'(0,-z))=2z.\]
If $\alpha_k>0$ is a Regge pole it follows then that:
\[ \tilde{\psi}'(0,\alpha_k)=-\frac{2\alpha_k}{\tilde{\psi}(0,-\alpha_k)}. \]
Note that the denominator cannot vanish as $\alpha_k\neq 0$.
Therefore the $m_k$ must also be identical.
\end{proof}
\begin{coro}
If two warped balls have the same Regge poles then the corresponding functions $F$ are identical.
\end{coro}

The relevance of the function $F$ for the inverse problem is that if $K$ is the Gelfand-Levitan transformation operator then it satisfies the so called basic integral equation:
\[F(x+y)+K(x+y) + \int_x^{\infty}K(x,t)F(t+u)\dd t=0, \quad (y\geq x) \tag{{\cite[Equation (4.5.8)]{Le2018}}}. \]
Based on the theory of this equation they show that:
\begin{thm}[{\cite[Theorem 4.7.1]{Le2018}}]
Let two potentials $q$ and $\tilde{q}$ satisfy $\int_0^\infty x|q(x)|\dd x<+\infty$ and assume that the corresponding functions $F_q,F_{\tilde{q}}$ are equal then:
$q=\tilde{q}$.
\end{thm}
The uniqueness result, Theorem~\ref{MainInverse} which we restate here, then follows directly:
\begin{thm} If two warped balls have the same Regge poles then the potentials $Q_f$ are identical. Assuming that the boundary values $f(0), f'(0)$ are also equal, it follows that the warped balls are defined by the same conformal factor $f$. \end{thm}

\subsection{A formula for the Weyl-Titchmarsh function}
We will now solve the inverse problem again following the strategy of~\cite{BKW2003}, under the Assumption~\ref{cf3}. Our aim is to choose a sequence $\gamma_n$ of contours that grow in the limit $n\to +\infty$ and such that for fixed $z$:
\[ \lim_{n\to+\infty} \frac{1}{2i\pi} \int_{\gamma_n}h_z(\mu)m(\mu)\dd \mu =0. \]
In the above, $h_z$ is an auxiliary function that we will choose after finding such contours. This will provide the formula stated in Theorem~\ref{thm:wt_expression}. 
The choice of contour is based on the study of the asymptotics of $m$ using the results of the previous section and the relation:
\begin{lemm}
\[ m(z)= m(-z) -\frac{2\sin(z\pi)}{\pi \psi(0,z)\psi(0,-z)}.\]
\end{lemm}
\begin{proof}
Recall from~\cite[\S 3.12,p.43, Equation~(2)]{watson1995treatise} that:
\[W(J_z(t), J_{-z}(t))=-\frac{2\sin{z\pi}}{\pi t}.\]
Thus: 
\[ W(\psi_0(x,z),\psi_0(x,-z))=W(\psi_0(0,z),\psi_0(0,-z))=\frac{2\sin z\pi}{\pi}. \]
Since $Q_f$ has support $[0,a]$, for any $x\geq a, z\in \C$, $\psi(x,z)=\psi_0(x,z)$, so that:
\[ W(\psi(0,z),\psi(0,-z))=W(\psi(x,z),\psi(x,-z))=\frac{2\sin z\pi}{\pi}.\]
The desired equality then follows from the definition of the Wronskian.
\end{proof}

The asymptotics of $m$ in the half-plane $\Re z \geq 0$ is well-known for a wide class of potentials, we can nevertheless find these asymptotics here directly.
Note that:
\[ \psi'(0,z)=-\sqrt{\lambda}J_z'(\sqrt{\lambda}) - K(0,0)J_z(\sqrt{\lambda}) +\int_0^\infty \partial_xK(0,s)J_z(\sqrt{\lambda}e^{-s})\textrm{d}s.\]
The integral is estimated as before and using Lemma~\ref{lemm:wt_bessel_asymp} we find that, on the half-plane $\Re z\geq 0$:
\[\psi'(0,z)=\gO{z\frac{\lambda^{\frac{z}{2}}}{2^{z}\Gamma(z+1)}}{|z|}{\infty}.\]
Observe now that it follows from Equation~\eqref{eq:jost_re_pos} combined with Rouché's lemma that, for large enough $|z|$, $\psi(0,z)$ does not vanish and hence:
\begin{equation} \label{eq:asymp_M_php} m(z) = \gO{|z|}{|z|}{\infty}, \ \Re z \geq 0.\end{equation}

Lemma~\ref{lemm:asymp_sector} enables us to extend this to sectors of the form: \[S_{\varepsilon,\delta}=\{\Re z < -\varepsilon |\Im z| \}\cap U_\delta, \varepsilon >0.\] 
Indeed, appealing once more to Rouché's lemma, Lemma~\ref{lemm:asymp_sector} shows that asymptotically there are no zeros in this region. 
Moreover, combining the asymptotics, one has:
\[ \begin{aligned} \psi(0,z)\psi(0,-z)&=\frac{\partial^p_xK(0,2a^{-})e^{-2az}}{z\Gamma(z)\Gamma(1-z)}\left(1+ \po{1}{|z|}{\infty} \right),\\&= \frac{\sin(\pi z)\partial^p_xK(0,2a^{-})e^{-2az}}{z\pi}\left(1+ \po{1}{|z|}{\infty} \right),\end{aligned} \]
where in the last equality we have used the Complement Formula for the Gamma function. Therefore, for $|z|$ large enough and $z\in S_{\varepsilon,\delta}$:

\begin{equation}\label{eq:asymp_M_sector} m(z) = m(-z) - \frac{2z^{p+3}e^{2az}}{\partial^p_xK(0,2a^{-})}\left(1+\po{1}{|z|}{\infty} \right)= m(-z)+ \po{1}{|z|}{\infty}. \end{equation}

We must now treat the regions near the imaginary axis defined by the inequalities:
\[ 0 \geq \Re z \geq -\varepsilon|\Im z|.\]
This will not be achieved uniformly but instead on specific circular arcs with increasing radius. Thanks to the asymptotics in Lemma~\ref{lemm:asymp_saut}, we can refer to~\cite[Lemma 6]{BKW2003}, which can be applied directly to show:
\begin{lemm}
\label{lemm:lemm6brown}
Let $\varepsilon >0$. There is a real $\tau$, such that for all sufficiently large $n\in \mathbb{N}$, \[ \psi(z,0)= \frac{\lambda^{\frac{z}{2}}}{2^z\Gamma(z+1)}g(z),\]
where $|g(z)| \geq \frac{1}{3}$ on all circular arcs given by $|z|=\frac{(2n\pi + \tau)}{2a}$ and $-\varepsilon |\Im z| \leq \Re z\leq 0$.
\end{lemm}

The previous lemmata give the required information to construct contours for our Cauchy integral argument, however first we shall pause to make a remark about non-positive integer zeros.
\begin{lemm}
\label{lemm:positive_integers_simple_zeros}
 Integers $n\in \mathbb{Z}$ are at most simple zeros of $z\mapsto \psi(0,z)$. Furthermore, if $z=n$ is a Regge pole, the residue of $m$ is then given by:
\[ \textrm{res}_{n}(m)=\frac{(-1)^n}{2\left(\frac{\dd}{\dd z}\psi(0,n)\right)^2} \]
\end{lemm}
\begin{proof}
Let us assume that $\psi(0,n)=0$ for some $n\in \mathbb{Z}$. Observe first that the Wronskian, $W(\psi(0,z),\psi(0,-z))$ is given by:
\[ W(\psi(0,z),\psi(0,-z)) =\frac{\sin{\pi z}}{\pi}.\]
Differentiating this identity respect to $z$ and evaluating in $z=n$, it is easily seen that:
\[ 2\psi'(0,-n)\frac{\dd}{\dd z}\psi(0,n) = \cos(\pi n)=(-1)^n.\]
So that: $\frac{\dd}{\dd z}\psi(0,n)\neq 0.$
\end{proof}

In the choice of contour we must of course pay attention to the domain on which we derived our asymptotics: $S_{\varepsilon,\delta}$. Let us fix $\varepsilon >0$, set $\delta = \frac{1}{4}$ and choose $\varepsilon > \varepsilon' > 0$.
 Let: \begin{itemize}\item  $\tau$ be the real given by Lemma~\ref{lemm:lemm6brown} and $N_0 \in \mathbb{N}$, be large enough such that the conclusion holds for any $n\geq N_0$. 
 \item $R_0 \in \mathbb{R}_+^*$, such that if $z \in \C_+\cap(\C \setminus \overline{B(0,R_0)})$, then:
 \[ \left \lvert \psi(0,z) - \frac{\lambda^{\frac{z}{2}}}{2^z\Gamma(z+1)}\right \rvert < \left \lvert \frac{\lambda^{\frac{z}{2}}}{2^z\Gamma(z+1)}\right \rvert.\]
 This is possible by Lemma~\ref{lemm:asymp_pos_hp}.
 \item $R_1 \in \mathbb{R}_+^*$, such that if $z\in S_{\varepsilon',\delta}$, then:
 \[\left\lvert \psi(0,z) - \frac{\partial^p_sK(0,2a^{-})\lambda^{\frac{z}{2}}e^{-2az}}{2^z\Gamma(z+1)z^{p+2}} \right\rvert < \left\lvert \frac{\partial^p_sK(0,2a^{-})\lambda^{\frac{z}{2}}e^{-2az}}{2^z\Gamma(z+1)z^{p+2}} \right\rvert.\]
 This is possible by Lemma~\ref{lemm:asymp_sector}.
\end{itemize}
\begin{figure}[h!]
\centering
\includegraphics[scale=.5]{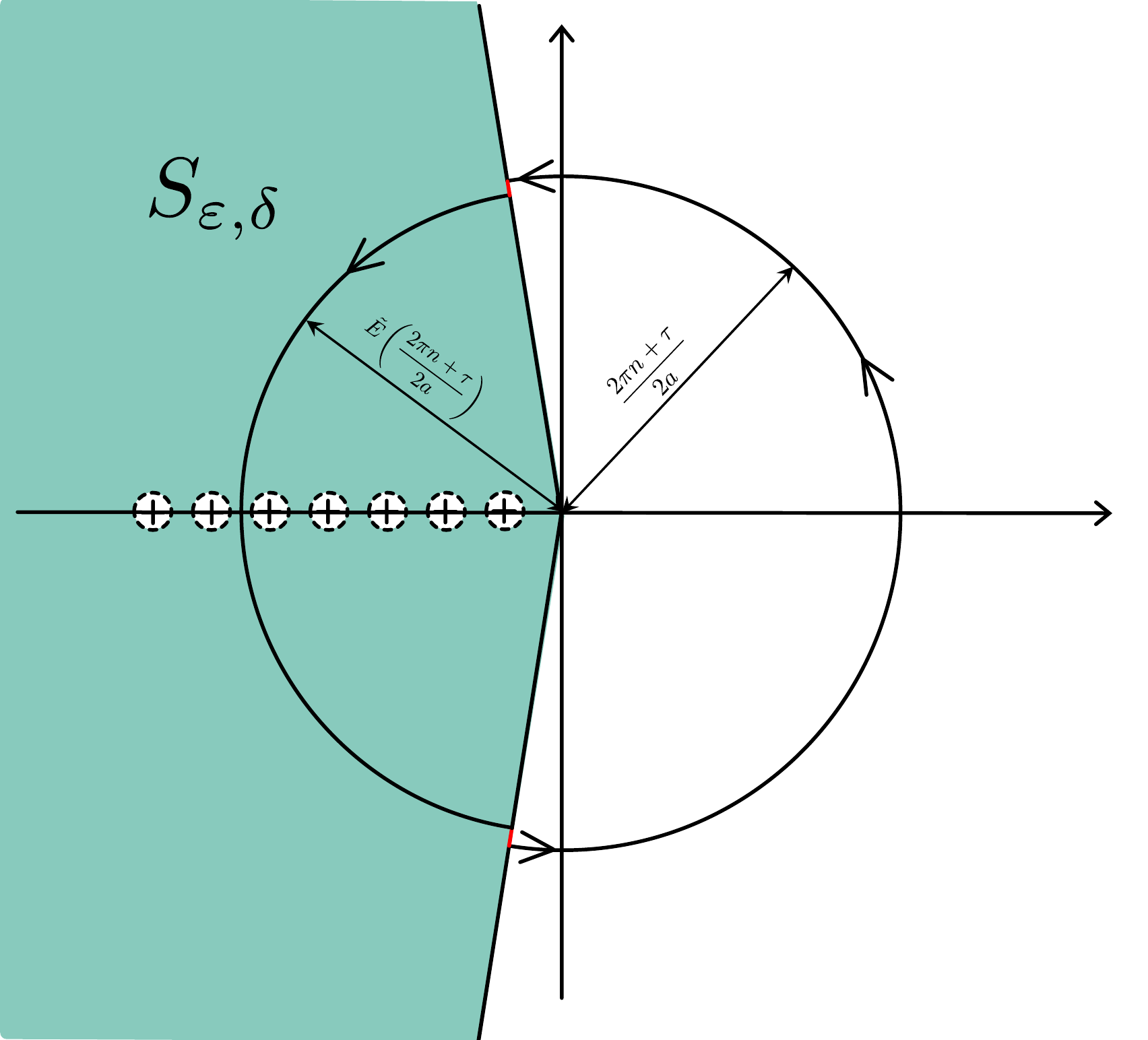}
\caption{\label{fig:contour}Sketch of $S_{\varepsilon,\delta}$ and a generic curve $\gamma_n$.}
\end{figure}

Let us denote $\tilde{E}(x)$ the half-integer closest to $x\in \mathbb{R}$; if $x$ is integer we choose $\tilde{E}(x)=x+\frac{1}{2}$.
 Now choose $N_1 \in \mathbb{N}$ such that for every $n\geq N_1$, we have: \[\min\left(\tilde{E}\left( \frac{2\pi n + \tau}{2a}\right),\frac{2\pi n + \tau}{2a}\right)\geq \max(N_0,R_0,R_1).\]
 
For any $n\geq N_1$, $\psi(0,z)$ none of the zeros of $\psi(0,z)$ are located on the circular arcs defined by:
\[ \Gamma_n^1=C\left(0,\tilde{E}\left(\frac{2\pi n + \tau}{2a} \right)\right)\cap S_{\varepsilon,\delta}, \quad \Gamma_n^2=C\left(0,\frac{2\pi n + \tau}{2a}\right)\cap S_{\varepsilon,\delta}^c.\]

Increasing if necessary $N_1$, we can also assume that there are no zeros on the line segments of $\Re z = -\varepsilon\Im z$ that can be used to join these two circular arcs as in Figure~\ref{fig:contour}, let us denote by $\gamma_n$ the obtained closed contour.

Based on our asymptotics for $z\mapsto m(z)$, the appropriate auxiliary function is: \[h_z(\mu)=\left(\frac{z}{\mu}\right)^{2}\frac{1}{z-\mu}.\] We shall now show that, as desired, for fixed $z\in \mathbb{C}$ such that $z$ is not a zero of $\psi(0,z)$:
 \[ \lim_{n\to+\infty} \frac{1}{2i\pi} \int_{\gamma_n}h_z(\mu)m(\mu)\dd \mu =0. \]
For large enough $n$, $z$ is completely enclosed by the contour. Denote by $s_n$ the union of the two line segments that join the circular arcs and write:
 \[\int_{\gamma_n}h_z(\mu)m(\mu)\dd \mu = \int_{s_n} h_z(\mu) m(\mu) \dd \mu + \int_{\Gamma^1_n}h_z(\mu)m(\mu)\dd \mu + \int_{\Gamma^2_n}h_z(\mu)m(\mu)\dd\mu.\]
$\Gamma^1_n$ and $s_n$ are completed contained in $S_{\varepsilon',\delta}$, hence using Equation~\eqref{eq:asymp_M_sector}, we see that their contribution vanishes when $n\to \infty$ by the mean value theorem.
Appealing to Lemma~\ref{lemm:lemm6brown}, it follows that on the circular arcs given by $|z|=\frac{2n\pi + \tau}{2a}$ and $-\varepsilon |\Im z| \leq \Re z \leq 0$, one has: 
\[m(z)= m(-z) + \po{1}{n}{\infty},\]
combining this with Equation~\eqref{eq:asymp_M_php} and the mean value theorem, it follows that the contribution along $\Gamma^2_n$ vanishes also in the limit $n\to \infty$.

Let us now assume for the moment that $\psi(0,0)\neq 0$. The Cauchy integral theorem applied for large enough $n$ such that $\gamma_n$ encloses $z$ leads to:
\[ \frac{1}{2\pi i}\int_{\gamma_n}h_z(\mu)m(\mu)\dd \mu = -m(z) + m(0)+zm'(0) + \sum_{\text{\parbox{2cm}{\center $z_i$ Regge pole enclosed by $\gamma_n$}}}\textrm{res}_{z_i}(h_zm). \]
Since the limit on the left exists and vanishes, it follows that the sum over all Regge poles converges and:
\begin{equation} \label{WT} 
m(z)= m(0) +zm'(0) + \sum_{\textrm{$z_i$ Regge pole}} \textrm{res}_{z_i}(h_zm). 
\end{equation}

\noindent If it so happens that $\psi(0,0)=0$, then this formula should be adjusted to:
\begin{equation} \label{WT0}
m(z) = g'(0) + z\frac{g''(0)}{2}+\frac{\textrm{res}_0(m)}{z} +\sum_{\underset{\textrm{Regge pole}}{z_i\neq0}} \textrm{res}_{z_i}(h_zm), 
\end{equation}
with $g(\mu)=\mu m(\mu).$ \\

This is in essence the content of Theorem~\ref{thm:wt_expression}:
\begin{proof}[Proof of Theorem~\ref{thm:wt_expression}]
If we assume, for simplicity, that the Regge poles are all simple, using the well-known asymptotics \cite{Si1999}:
$$
 M(-z^2) = -z  + o(1), \quad z \to +\infty, \ z \in \R,   
$$
we obtain from either Equation~\eqref{WT0} or \eqref{WT}  a synthetic expression for the Weyl-Titchmarsh function:
$$
M(-z^2) = -z + \sum_{\underset{\textrm{Regge pole}}{z_i}}  \frac{a_i}{z - z_i}. 
$$
\end{proof}

In order to obtain the desired uniqueness result from our current strategy, we must now argue that the terms in this expansions can all be determined either by known asymptotics for the Weyl-Titchmarsh function $m$, which is the case for the first order polynomial, or directly from the Regge poles. As an intermediate step, let us first explain that the residues can all be determined from the Jost function. If $z_k \neq 0$ is a pole of order $j$ then:
\[ \textrm{res}_{z_k}(h_zm)=\left.\frac{\dd^{j-1}}{\dd \mu^{j-1}} \left(\frac{(\mu-z_k)^j}{z-\mu}\left(\frac{z}{\mu}\right)^2m(\mu) \right)\right|_{\mu=z_k}. \]
Since $m(\mu)=\frac{\psi'(0,\mu)}{\psi(0,\mu)}$, the only quantities appearing in developing this expression that have yet to be determined are that of $\psi'(0,z)$ and its derivatives with respect to $z$, evaluated at the Regge pole $z_j$.  For any possible non-zero integer poles, we have already seen in Lemma~\ref{lemm:positive_integers_simple_zeros} that they are at most simples poles and that the residue of $m$ is completely determined if we know the Jost function. For non-integer poles, the required values can be determined iteratively in the same way by differentiating the Wronskian as many times as necessary. This method works since $\sin$ only vanishes on the integers, so that $\psi(0,z)$ and $\psi(0,-z)$ cannot vanish simultaneously for non-integer values.
For definiteness let us illustrate this for a non-integer Regge pole $z_k$ of order $2$. Recall the following identity:
\[ W(\psi(0,z),\psi(0,-z))=\psi(0,z)\psi'(0,-z)-\psi'(0,z)\psi(0,-z)=\frac{2\sin z\pi}{\pi},\]
When evaluating at $z_k$, we obtain:
\[-\psi'(0,z_k)\psi(0,-z_k)=\frac{2\sin z_k\pi}{\pi}\neq 0, \]
from which we can determine the value of $\psi'(0,z_k)$. Proceeding identically after differentiating with respect to $z$ (denoted by $\dot{}$~) we obtain:
\[ -\dot{\psi}'(0,z_k)\psi(0,-z_k) +\psi'(0,z_k)\dot{\psi}(0,-z_k)=\cos z_k\pi.\]
Since $\psi(0,-z_k)$ does not vanish and $\dot{\psi}(0,-z_k)$ can be calculated if the Jost function is known, we can infer the value of $\dot{\psi}'(0,z_k)$.
When the pole is of higher order, we can repeat this procedure iteratively to calculate all the $z$-derivatives of $\psi'(0,z)$ up to the required order.

This argument show that $m$ is completely determined by the Regge poles if they determine uniquely the Jost function. However, this was established in Proposition~\ref{prop:jost_unique}. We can now give the second proof of Theorem~\ref{MainInverse} :
\begin{proof}[Proof of Theorem~\ref{MainInverse}]
It follows from formulae \eqref{WT} and \eqref{WT0} and the above discussion that the Weyl-Titchmarsh function $M$ is uniquely determined by the set of Regge poles $z_i$. Appealing to the well-known Borg-Marchenko theorem~\cite{Si1999, GeSi2000, Be2001}, it follows that the potential $-\lambda e^{-2x} + Q_f(x)$ is uniquely determined on $\R^+$. Finally, using the explicit form of the potential $Q_f$, and assuming $f(0),f'(0)$ known, a simple argument using uniqueness of Cauchy for second order ODEs shows that the warping function $f$ is also uniquely determined.  	
\end{proof}




\printbibliography[title=Bibliography]

 \end{document}